\documentclass[reqno]{amsart}
\usepackage[utf8]{inputenc}
\usepackage{amsthm}
\usepackage{amssymb}
\usepackage{graphicx}
\usepackage{enumerate}
\usepackage{tikz}
\usepackage[boxed,lined]{algorithm2e}
\usepackage{hhline}
\usepackage{seqsplit}
\usepackage{braket}
\usepackage{todonotes}
\usepackage{url}
\usepackage{caption}
\usetikzlibrary{decorations.pathreplacing}
\usetikzlibrary{quantikz2}

\usepackage[ocgcolorlinks, linkcolor=red!90!black, citecolor=blue!80!black, pagebackref=true]{hyperref}

\newtheorem{theorem}{Theorem}[section]
\newtheorem*{theorem*}{Theorem}
\newtheorem{lemma}[theorem]{Lemma}
\newtheorem{proposition}[theorem]{Proposition}
\newtheorem{definition}[theorem]{Definition}
\newtheorem{corollary}[theorem]{Corollary}

\theoremstyle{definition} \newtheorem{example}[theorem]{Example}

\newcommand{\CC}{\mathbb{C}}

\newcommand{\HH}{\mathbb{H}}
\newcommand{\RR}{\mathbb{R}}

\newcommand{\ZZ}{\mathbb{Z}}

\newcommand{\PSU}{\operatorname{PSU}}

\newcommand{\U}{\operatorname{U}}
\newcommand{\SU}{\operatorname{SU}}
\newcommand{\SO}{\operatorname{SO}}
\newcommand{\PSO}{\operatorname{PSO}}
\newcommand{\Sp}{\operatorname{Sp}}
\renewcommand{\O}{\operatorname{O}}
\newcommand{\so}{\mathfrak{so}}
\newcommand{\su}{\mathfrak{su}}
\renewcommand{\sp}{\mathfrak{sp}}

\newcommand{\tr}{\operatorname{tr}}

\newcommand{\diag}{\operatorname{diag}}
\newcommand{\sgn}{\operatorname{sgn}}
\newcommand{\vspan}{\operatorname{span}}

\renewcommand{\Re}{\operatorname{Re}}
\renewcommand{\Im}{\operatorname{Im}}

\DeclareMathOperator{\Ad}{Ad}
\DeclareMathOperator{\ad}{ad}
\DeclareMathOperator{\id}{id}
\DeclareMathOperator{\Aut}{Aut}
\DeclareMathOperator{\Out}{Out}
\DeclareMathOperator{\Inn}{Inn}

\renewcommand{\P}{\operatorname{P}}
\newcommand{\secsymbol}{§}
\renewcommand{\S}{\operatorname{S}}

\newcommand{\p}{\mathfrak{p}}
\newcommand{\G}{\mathcal{G}}
\newcommand{\g}{\mathfrak{g}}
\renewcommand{\H}{\mathcal{H}}
\newcommand{\h}{\mathfrak{h}}
\newcommand{\K}{\mathcal{K}}
\renewcommand{\k}{\mathfrak{k}}
\newcommand{\A}{\mathcal{A}}
\newcommand{\B}{\mathcal{B}}
\renewcommand{\a}{\mathfrak{a}}
\newcommand{\T}{\mathcal{T}}

\newcommand{\M}{\mathcal{M}}
\newcommand{\Q}{\mathcal{Q}}
\newcommand{\Z}{\mathcal{Z}}

\begin{document}

\title{Real 3-qubit gate decompositions via triality}
\author{Brendan Pawlowski}
\date{}

\begin{abstract} We show that any unimodular real 3-qubit gate can be expressed as the product of at most 14 CNOT gates plus single-qubit gates, improving on the bound of 16 due to Wei and Di. Our method uses the exotic \emph{triality symmetry} of $\PSO(8)$, and we explore some of the useful properties of this map in relation to the study of real 3-qubit gates.
\end{abstract}

\maketitle

\section{Introduction}

Any quantum gate can be written as a product of single-qubit gates and 2-qubit gates \cite{divincenzo-2q}. In fact, almost any 2-qubit gate $V \in \U(4)$ is universal, meaning that every $n$-qubit gate can be expressed as a quantum circuit involving only single-qubit gates and copies of $V$ acting on pairs of qubits \cite{2q-universal-1995, 2q-universal-lloyd}. For the sake of comparing different gate decomposition algorithms, it is convenient to fix a particular $V$; a popular choice is the controlled-NOT (CNOT) gate, shown to be universal in \cite{elementary-gates}.

Once the existence of such quantum circuits has been established, it is natural to ask how large they must be: for instance, what is the minimal $N$ such that every $n$-qubit gate $V \in \U(2^n)$ can be written as $L_1 C_1 L_2 C_2 \cdots L_N C_N L_{N+1}$ where the $C_i$ are CNOTs and the $L_i$ are in $\U(2)^{\otimes n}$? In almost all cases the exact answer is unknown. Shende, Markov, and Bullock \cite{shende-markov-bullock} showed using a dimension-counting argument that at least $(4^n-3n-1)/4$ CNOTs are required. Many decomposition algorithms have been devised providing upper bounds; the current best seems to be $\tfrac{22}{48} 4^n - \tfrac{3}{2} 2^n + \tfrac{5}{3}$ CNOTs due to Krol and Al-Ars \cite{block-ZXZ}, who also give a comprehensive overview of these algorithms.

The case $n=2$ is well-understood, at least. Any $V \in \U(4)$ can be written as the product of at most 3 CNOTs plus single-qubit gates, this bound is sharp, and there are exact algorithms to compute optimal decompositions \cite{2q-optimal}. A special feature of this case that simplifies analysis is the ``magic basis'' coordinate change, which conjugates the subgroup $\SU(2) \otimes \SU(2) \subseteq \SU(4)$ generated by single-qubit gates onto $\SO(4)$. This allows the technique of Cartan decomposition to be applied directly, which by contrast cannot be done for the subgroup $\SU(2)^{\otimes n} \subseteq \SU(2^n)$ when $n > 2$.

In this paper we study the special case of unimodular ($\det = 1$) real 3-qubit gates. Wei and Di \cite{wei-di} showed that any $V \in \SO(8)$ is the product of at most 16 CNOTs plus single-qubit gates. Our main result (Theorem~\ref{thm:main}) is an explicit quantum circuit realizing any element of $\SO(8)$ as the product of at most 14 CNOTs plus 35 single-qubit rotations.

Perhaps more interesting than this modest improvement is our technique, which relies on the \emph{triality symmetry} $\T : \PSO(8) \to \PSO(8)$.  This is a rather exotic map which only appears in association with the $\so(8)$ Lie algebra, so we spend some time defining it and exploring its properties.  Its utility in the study of 3-qubit gates does not seem to have been observed before; in some ways it behaves as a real 3-qubit version of the magic basis transformation which is so useful in the 2-qubit case. More specifically, it transforms many subgroups defined in terms of tensor products into straightforward block matrix subgroups, which makes matrix factorizations easier to see.

In Section~\ref{sec:prelim} we set up notation and review some Lie algebra and Lie group material, including Cartan decomposition. Section~\ref{sec:triality} is devoted to the triality map. In Section~\ref{sec:real-circuits} we discuss an interesting Cartan decomposition of $\PSO(8)$ involving triality, which we apply in Section~\ref{sec:pso8-cartan} to construct our new circuit.

\section{Preliminaries} \label{sec:prelim}
\subsection{Matrix notation}
We use the interval notation $[m,n]$ for the set $\{m,m{+}1,\ldots,n\}$ when $m,n$ are integers. As a special case, $[n] = [1,n] = \{1,2,\ldots,n\}$. If $A$ is an $m \times n$ matrix and $I \subseteq [m], J \subseteq [n]$ are subsets, write $A_{IJ}$ for the $|I| \times |J|$ submatrix $[A_{ij}]_{i \in I, j \in J}$.  For instance, if $A = \left[ \begin{smallmatrix} 1 & 2 & 3 \\ 4 & 5 & 6 \\ 7 & 8 & 9 \end{smallmatrix} \right]$ then $A_{[2], [1,3]} = \left[ \begin{smallmatrix} 1 & 3 \\ 4 & 6 \end{smallmatrix} \right]$.

Given square matrices $A_1, \ldots, A_k$, write $A_1 \oplus \cdots \oplus A_k$ for the block-diagonal matrix with $A_1, \ldots, A_k$ as blocks. For example,
\begin{equation*}
\begin{bmatrix} 1 & 2 \\ 3 & 4 \end{bmatrix} \oplus \begin{bmatrix} 5 & 6 \\ 7 & 8 \end{bmatrix} = \begin{bmatrix} 1 & 2 & 0 & 0 \\ 3 & 4 & 0 & 0 \\ 0 & 0 & 5 & 6 \\ 0 & 0 & 7 & 8 \end{bmatrix} 
\end{equation*}
More generally, if $S_1, \ldots, S_k$ are sets of matrices, then
\begin{equation*}
S_1 \oplus \cdots \oplus S_k = \{A_1 \oplus \cdots \oplus A_k : A_i \in S_i \text{ for each $i$}\}.
\end{equation*}
We also write $\diag(x_1,\ldots,x_n)$ for the diagonal matrix with diagonal entries $x_1,\ldots,x_n$. As a special case that will arise frequently, for each subset $J \subseteq [n]$ define $\Delta_{J}$ as the diagonal matrix with
\begin{equation*}
(\Delta_J)_{ii} = \begin{cases} -1 & \text{if $i \in J$} \\
						   1 & \text{if $i \notin J$} \end{cases}
\end{equation*}
Finally, the operator $\S$ applied to a matrix group gives the subgroup where $\det = 1$, and $\P$ takes the quotient modulo scalar multiplication. 

\begin{example} \hfill
\begin{enumerate}[(a)]
\item $\Delta_{\{2,3\} \subseteq [5]} = \diag(1,-1,-1,1,1)$.
\item $\SO(2) \oplus \SO(2) = \{A \oplus B : A,B \in \O(2), \det A = \det B = 1\}$.
\item $\S(\O(2) \oplus \O(2))$ is (a) plus the connected component where $\det A = \det B = -1$.
\item $\P\S(\O(2) \oplus \O(2))$ is the image of (b) in $\PSO(4)$.
\end{enumerate}
\end{example}

We are especially concerned here with $\PSO(8)$, the group of real orthogonal matrices of determinant 1, subject to the rule that $A$ and $-A$ are considered the same matrix. The reason for making the distinction between $\PSO(8)$ and $\SO(8)$ is that the triality map involves an indeterminacy of sign, and so is well-defined as a map $\PSO(8) \to \PSO(8)$ but not $\SO(8) \to \SO(8)$. That said, our final circuit decomposition will be perfectly valid in $\SO(8)$.

\subsection{Quantum gate conventions}
Take the Pauli matrices to be
\begin{equation*}
\sigma_x = \begin{bmatrix} 0 & 1 \\ 1 & 0 \end{bmatrix} \qquad \sigma_y = \begin{bmatrix} 0 & -i \\ i & 0 \end{bmatrix} \qquad \sigma_z = \begin{bmatrix} 1 & 0 \\ 0 & -1 \end{bmatrix}
\end{equation*}
and set $R_a(\theta) = \exp(-i\theta\sigma_a/2)$ for $a \in \{x,y,z\}$. These rotations generate $\SU(2)$ since $i\sigma_x, i\sigma_y, i\sigma_z$ generate $\su(2)$ as a Lie algebra, and it is well-known that any element of $\SU(2)$ can be written as $R_{a_1}(\theta_1)R_{a_2}(\theta_2)R_{a_3}(\theta_3)$ where $a_1 a_2 a_3$ is any word on the alphabet $\{x,y,z\}$ with no two consecutive letters the same. Our convention is that $\su(n)$ consists of the $n \times n$ skew-Hermitian matrices.

We take $\CC^2$ to have ordered orthonormal basis $\bra{0}, \bra{1}$, and more generally $(\CC^2)^{\otimes n}$ to have basis $\{\bra{w} : w \in \{0,1\}^n\}$, indexed by the length $n$ binary words ordered lexicographically. Write $C_i^j \in \U(2^n)$ for the CNOT gate with control qubit $i$ and target qubit $j$, i.e.\ the linear operator defined by
\begin{equation*}
\bra{w_1 \cdots w_n} \mapsto \begin{cases}
\bra{w_1 \cdots w_n} & \text{if $w_i = 0$}\\
\bra{w_1 \cdots \tilde{w}_j \cdots w_n} & \text{if $w_i = 1$}
\end{cases}
\end{equation*}
where $\tilde{x} = 1-x$. For instance, with $n = 2$,
\begin{equation*}
C_1^2 = \begin{quantikz}
& \ctrl{1} & \\
& \targ{}  &
\end{quantikz} = \left[\begin{smallmatrix} 1 & 0 & 0 & 0 \\ 0 & 1 & 0 & 0 \\ 0 & 0 & 0 & 1 \\ 0 & 0 & 1 & 0 \end{smallmatrix}\right]
 \qquad \text{and} \qquad C_2^1 = \begin{quantikz}
& \targ{} & \\
& \ctrl{-1}  &
\end{quantikz}  = \left[\begin{smallmatrix} 1 & 0 & 0 & 0 \\ 0 & 0 & 0 & 1 \\ 0 & 0 & 1 & 0 \\ 0 & 1 & 0 & 0 \end{smallmatrix}\right]
\end{equation*}

\subsection{Lie group notation and Cartan decomposition} \label{subsec:cartan}
The following notation will hold for the rest of the paper, unless otherwise noted.
\begin{itemize}
\item $\G$ is a compact connected Lie group with Lie algebra $\g$.
\item $\H_0$ is the connected component of the identity in a Lie group $\H$.
\item $\theta : \G \to \G$ is an automorphism and an involution, i.e.\ an invertible function satisfying $\theta(gh) = \theta(g)\theta(h)$ and $\theta(\theta(g)) = g$ for all $g,h \in \G$.
\item $c_g : \H \to \H$ is the conjugation automorphism $c_g(h) = ghg^{-1}$.
\item $\G^{\pm \theta} = \{g \in \G : \theta(g) = g^{\pm 1}\}$, so $\G^\theta$ is the subgroup of fixed points of $\theta$, while $\G^{-\theta}$ is just a subset. 
\item $\Z(S) = \{h \in \H : sh = hs \text{ for all $s \in S$}\}$ is the centralizer of a subset $S \subseteq \H$. We write $\Z(g)$ instead of $\Z(\{g\})$.
\item $\k$ and $\p$ are the $1$- and $(-1)$-eigenspaces of the induced Lie algebra homomorphism $d\theta : \g \to \g$, so $\k$ is a subalgebra and $\g = \k \oplus \p$. 
\item $\K = \exp \k = (\G^\theta)_0$ is the connected component of the identity in $\G^\theta$. We will call any closed subgroup arising this way a \emph{Cartan subgroup}.
\item $\a$ is a maximal abelian subalgebra of $\p$ and $\A = \exp(\a)$.
\item Later we will have more than one involution $\theta$, in which case we write $\K_\theta$, $\a_\theta$, etc.\ to be more specific.
\end{itemize}
We try to consistently use calligraphic capitals for abstract Lie groups, lowercase Fraktur for Lie algebras, capitals for matrices and named matrix groups, Greek letters for automorphisms of groups and Lie algebras, and lowercase $g,h,k$ etc.\ for elements of abstract Lie groups. That said, we also use calligraphic capitals for certain distinguished matrices and for the triality automorphism $\T$.

\begin{theorem}[Cartan decomposition] \cite[Ch.\ V, Theorem 6.7]{helgason} \label{thm:cartan}
With notation as above,  we have $\G = \K\A\K$. That is, any $g \in \G$ can be written as $k_1 a k_2$ with $k_1,k_2 \in \K$ and $a \in \A$. Also, $\g = \k \oplus \p$ and 
\begin{equation*}
\exp(\p) = \bigcup_{k \in \K} k\A k^{-1} \qquad \text{and} \qquad \p = \bigcup_{k \in \K} \Ad_k(\a)
\end{equation*}
where $\Ad_k : \g \to \g$ is the derivative of the conjugation map $c_k$ at the identity.
\end{theorem}

\begin{example} \label{ex:SU-conj}
If $\G = \U(n)$ and $\theta : \G \to \G$ is complex conjugation, then $\G^\theta = \O(n)$ and $\K = \SO(n)$ and $\G^{-\theta}$ is the set of symmetric unitary matrices.  At the Lie algebra level,
\begin{equation*}
\k = \so(n) \qquad \text{and} \qquad \p = \{\text{symmetric imaginary matrices}\}.
\end{equation*}
The choice of $\a$ is not unique, but we can take $\a \subseteq \p$ to be the imaginary diagonal matrices, making $\A = \exp(\a)$ the group of diagonal unitary matrices.  Hence the Cartan decomposition factors any unitary as $O_1 D O_2$ with $O_1,O_2 \in \SO(n)$ and $D$ unitary diagonal.
\end{example}

Besides the decomposition $\G = \K\A\K$, we will also use a basic corollary of Theorem~\ref{thm:cartan}.
\begin{corollary} \label{cor:ka}
$\g$ is generated as a Lie algebra by $\k$ and $\a$.
\end{corollary}

\begin{proof}
It is a basic fact in Lie theory that if $\ad_X : \g \to \g$ is the linear map $\ad_X(Y) = [X,Y]$, then $\Ad_{\exp(X)} = \exp(\ad_X)$ \cite[Ch.\ II, \secsymbol 5]{helgason}. That is,
\begin{equation*}
\Ad_{\exp(X)}(Y) = Y + [X,Y] + \frac{1}{2!}[X,[X,Y]] + \frac{1}{3!}[X,[X,[X,Y]]] + \cdots
\end{equation*}
In particular, this shows that $\Ad_{\exp(X)}(Y)$ is in the subalgebra generated by $X$ and $Y$. Theorem~\ref{thm:cartan} therefore implies that $\p$ is in the subalgebra generated by $\k$ and $\a$, which proves the corollary since $\g = \k \oplus \p$.
\end{proof}

How does one compute a Cartan decomposition of $g \in \G$ in practice? One approach is to compute the \emph{Cartan double} $g\theta(g)^{-1}$ of $g \in \G$. Indeed, if $g = k_1 a k_2$ with $k_1, k_2 \in \K$ and $a \in \A \subseteq \G^{-\theta}$, then
\begin{equation} \label{eq:cartan-double}
g\theta(g)^{-1} = (k_1 a k_2)(k_1 a^{-1} k_2)^{-1} = k_1 a^2 k_1^{-1}.
\end{equation}
If $\G$ is a matrix group (a subgroup of a finite quotient $\G'$ of $\U(n)$), then $\A$ is a subgroup of a maximal torus in $\G'$, which in turn must be conjugate to the torus of diagonal matrices \cite[\secsymbol 16.4]{humphreys}. Thus, we can compute $k_1$ and $a^2$ as in \eqref{eq:cartan-double} by diagonalizing $g\theta(g)^{-1}$ in an appropriate basis. There are finitely many square roots of $a^2$ in $\A$ (because the same is true for diagonal matrices), and each possible square root $a$ can be checked for correctness by seeing whether the putative $k_2 = a^{-1} k_1^{-1} g$ actually lies in $\mathcal{K}$.  With more work, one can choose a distinguished subset $\A_\circ \subseteq \A$ such that $g \in \K a \K$ for a \emph{unique} $a \in \A_\circ$, thereby avoiding this last issue; see \cite[Ch.\ VII, Theorem 8.6]{helgason} for the simply connected case.

For instance, in the setting of Example~\ref{ex:SU-conj} we would find a basis of real eigenvectors for the symmetric unitary matrix $U\overline{U}^{-1} = UU^T = O_1 D^2 O_1^{-1}$ in order to find $O_1$ and $D^2$. The next example covers the Cartan decompositions we will be concerned with.
\begin{example} \label{ex:SO-block}
Say $\G = \PSO(n)$ and $0 < p < n$. Set $\theta = c_{\Delta_{[p]}}$, so $\theta(U) = \Delta_{[p]}U\Delta_{[p]}$.  Then $\G^\theta = \Z(\Delta_{[p]})$ is the block-diagonal subgroup $\P\S(\O(p) \oplus \O(n-p))$, and the connected component of the identity is $\K = \P(\SO(p) \oplus \SO(n-p))$. Note that $\K$ is not necessarily isomorphic to $\PSO(p) \times \PSO(n-p)$.  For instance, $\P(\SO(4) \oplus \SO(4))$ has center $\{I,  \Delta_{[4]}\}$ while $\PSO(4) \times \PSO(4)$ has trivial center.

The Lie algebra $\so(n)$ of $\G$ consists of all $n \times n$ real skew-symmetric matrices. The $1$-eigenspace of $d\theta$ is $\so(p) \oplus \so(n-p)$, while the $(-1)$-eigenspace $\p$ consists of the matrices $\left[ \begin{smallmatrix} 0 & -M^T \\ M & 0 \end{smallmatrix}\right]$ where $M$ is $(n-p) \times p$. We take
\begin{align*}
&\a = \left\{ \begin{bmatrix}
0_{p,p} & 0_{p,n-2p} & -D\\
0_{n-2p,p} & 0_{n-2p,n-2p} & 0_{n-2p,p}\\
D & 0_{p,n-2p} & 0_{p,p}
\end{bmatrix} : \text{$D$ $p \times p$ real diagonal} \right\}\\
&\A = \left\{ \begin{bmatrix}
\cos D & 0_{p,n-2p} & -\sin D\\
0_{n-2p,p} & I_{n-2p,n-2p} & 0_{n-2p,p}\\
\sin D & 0_{p,n-2p} & \cos D
\end{bmatrix} : \text{$D$ $p \times p$ real diagonal} \right\} 
\end{align*}
For instance, Cartan decomposition in the case $n = 8$ and $p = 3$ says that every $V \in \PSO(8)$ can be factored as a product
\begin{equation*}
\left[\begin{smallmatrix}
\ast & \ast & \ast & 0 & 0 & 0 & 0 & 0\\
\ast & \ast & \ast & 0 & 0 & 0 & 0 & 0\\
\ast & \ast & \ast & 0 & 0 & 0 & 0 & 0\\
0 & 0 & 0 & \ast & \ast & \ast & \ast & \ast\\
0 & 0 & 0 & \ast & \ast & \ast & \ast & \ast\\
0 & 0 & 0 & \ast & \ast & \ast & \ast & \ast\\
0 & 0 & 0 & \ast & \ast & \ast & \ast & \ast\\
0 & 0 & 0 & \ast & \ast & \ast & \ast & \ast
\end{smallmatrix}\right]
\left[\begin{smallmatrix}
x_1 & 0 & 0 & 0 & 0 & -y_1 & 0 & 0\\
0 & x_2 & 0 & 0 & 0 & 0 & -y_2 & 0\\
0 & 0 & x_3 & 0 & 0 & 0 & 0 & -y_3\\
0 & 0 & 0 & 1 & 0 & 0 & 0 & 0\\
0 & 0 & 0 & 0 & 1 & 0 & 0 & 0\\
y_1& 0 & 0 & 0 & 0 & x_1 & 0 & 0\\
0 & y_2 & 0 & 0 & 0 & 0 & x_2 & 0\\
0 & 0 & y_3 & 0 & 0 & 0 & 0 & x_3
\end{smallmatrix}\right]
\left[\begin{smallmatrix}
\ast & \ast & \ast & 0 & 0 & 0 & 0 & 0\\
\ast & \ast & \ast & 0 & 0 & 0 & 0 & 0\\
\ast & \ast & \ast & 0 & 0 & 0 & 0 & 0\\
0 & 0 & 0 & \ast & \ast & \ast & \ast & \ast\\
0 & 0 & 0 & \ast & \ast & \ast & \ast & \ast\\
0 & 0 & 0 & \ast & \ast & \ast & \ast & \ast\\
0 & 0 & 0 & \ast & \ast & \ast & \ast & \ast\\
0 & 0 & 0 & \ast & \ast & \ast & \ast & \ast
\end{smallmatrix}\right]
\end{equation*}
\end{example}

\begin{definition} \label{def:canonical-params}
The \emph{$\K$-double coset} of $g \in \G$ is the set
\begin{equation*}
\K g \K = \{k_1 g k_2 : k_1,k_2 \in \K\}.
\end{equation*}
By the \emph{canonical parameters} of $g \in \G$ (with respect to the subgroup $\K$), we mean any data associated to $g$ that uniquely identifies its $\K$-double coset $\K g \K$.
\end{definition} 

\begin{example} \label{ex:U2xU2}
Say $\G = \U(4)$ and $\theta = c_{\Delta_{[2]}}$ is conjugation by $\diag(-1,-1,1,1)$. Then $\K = \G^\theta = \U(2) \oplus \U(2)$. We claim that the singular values of the upper-left $2 \times 2$ corner $U_{[2][2]}$ of $U \in \U(4)$ are canonical parameters with respect to $\K$.

On the one hand, left- or right-multiplying by an element of $\K$ does not change these singular values. On the other, Cartan decomposition tells us that any $U \in \U(4)$ can be factored as
\begin{equation} \label{ex:U4-cartan}
\begin{bmatrix} V_1 & 0 \\ 0 & V_2 \end{bmatrix}
\begin{bmatrix} \cos x_1 	& 	0	     & -\sin x_1 & 0 \\   
				0  			& \cos x_2	 & 0         & -\sin x_2\\
				\sin x_1  	& 0			 & \cos x_1  & 0\\
				0  			& \sin x_2	 & 0         & \cos x_2
\end{bmatrix}
\begin{bmatrix} V_3 & 0 \\ 0 & V_4 \end{bmatrix}
\end{equation}
with $V_i \in \U(2)$. The middle matrix can be freely left- or right-multiplied by any matrix $\diag(\pm 1, \pm 1, \pm 1, \pm 1)$ since these lie in $\K$, by which means we can make $\cos x_i, \sin x_i \geq 0$. Then $U_{[2][2]} = V_1 \diag(\cos x_1, \cos x_2) V_3$ has singular values $\cos x_1, \cos x_2$. These uniquely determine $\sin x_1, \sin x_2 \geq 0$, and therefore the whole middle matrix in \eqref{ex:U4-cartan} and the $\K$-double coset of $U$. This completes the argument.
\end{example}

Note that the sets $\k$, $\p$, $\K$ in Theorem~\ref{thm:cartan} are uniquely determined by $\theta$, while $\a$ and $\A$ are not. It may be convenient to replace $\A$ in the Cartan decomposition $\G = \K\A\K$ with a different set $\A' \subseteq \G$ which is not necessarily a torus. By definition, $\G = \K\A'\K$ holds if $\A'$ covers all possible canonical parameters.

\begin{example}
Continuing on from Example~\ref{ex:U2xU2}, we claim $\K C_2^1 \K C_2^1 \K = \U(4)$. Set $\A' = C_2^1 \K C_2^1$. We must show that for any $\sigma_1,\sigma_2 \in [0,1]$, there must exist $A \in \A'$ such that $A_{[2][2]}$ has singular values $\sigma_1, \sigma_2$.     The canonical parameters of $C_2^1 \diag(V_1,V_2) C_2^1 \in \A'$ are the singular values of its upper-left $2 \times 2$ corner
\begin{equation} \label{eq:corner}
\begin{bmatrix} 1 & 0 \\ 0 & 0 \end{bmatrix} V_1 \begin{bmatrix} 1 & 0 \\ 0 & 0 \end{bmatrix} + \begin{bmatrix} 0 & 0 \\ 0 & 1 \end{bmatrix} V_2 \begin{bmatrix} 0 & 0 \\ 0 & 1 \end{bmatrix}.
\end{equation}
By choosing $V_i = \left[ \begin{smallmatrix} \cos x_i & -\sin x_i \\ \sin x_i & \cos x_i \end{smallmatrix} \right]$ for $i = 1,2$ we can make the matrix \eqref{eq:corner} have any arbitrary singular values $\cos x_1, \cos x_2$.
\end{example}

As explained above, if $\K g \K = \K h \K$ then the Cartan doubles $g\theta(g)^{-1}$ and $h\theta(h)^{-1}$ are conjugate. If $\G$ is simply connected, this is a necessary and sufficient condition to have $\K g \K = \K h \K$, so the conjugacy class of the Cartan double $g\theta(g)^{-1}$ supplies canonical parameters \cite[Ch.\ VII, Theorem 8.6]{helgason}. For instance, the $\SO(n)$-double coset of $U \in \SU(n)$ is uniquely identified by the eigenvalues of $UU^T$.

Unfortunately neither $\SO(n)$ nor $\PSO(n)$ is simply connected for $n > 1$. We will not attempt to fix this issue in general, contenting ourselves with definitions that work in our specific cases of interest, which we return to in \secsymbol\ref{subsec:block-subgroups}.

We end this section with a simple but useful lemma on Lie groups.
\begin{lemma} \label{lem:gp-eq} If $\H_1$ is a connected Lie group and $\H_2 \subseteq \H_1$ is a closed subgroup with $\dim \H_2 = \dim \H_1$, then $\H_1 = \H_2$. 
\end{lemma}
\begin{proof}
It is well-known that $\H_1/\H_2$ can be made into a smooth manifold of dimension $\dim \H_1 - \dim \H_2$ so that the quotient map $\pi : \H_1 \to \H_1/\H_2$ is smooth. In the present case, $\H_1/\H_2$ must be discrete since it is a manifold of dimension 0. Therefore the disjoint union $\H_1 = \bigsqcup_{x \in \H_1/\H_2} \pi^{-1}(x)$ is a disjoint union of open sets. Since $\H_1$ is connected, there can only be one set in this union, i.e.\ $\H_1 = \H_2$.
\end{proof}

\subsection{Commuting involutions} \label{subsec:commuting-cartan}
Suppose $\theta_1, \theta_2 : \G \to \G$ are two commuting involutive automorphisms, i.e.
\begin{equation*}
\theta_1^2 = \theta_2^2 = \id = \theta_1 \theta_2 \theta_1^{-1} \theta_2^{-1}.
\end{equation*}
Apply Cartan decomposition with respect to $\theta_1$ to $g \in \G$, writing $g = k_1 a  k_2$ with $k_1, k_2 \in \K_{\theta_1}$ and $a_1 \in \A_{\theta_1}$. The commuting property implies $\theta_2(\K_{\theta_1}) \subseteq \K_{\theta_1}$, so one can Cartan decompose $k_1, k_2$ using the automorphism $\theta_2 : \K_{\theta_1} \to \K_{\theta_1}$. The relevant fixed-point subgroup in $\K_{\theta_1}$ is $(\K_{\theta_1} \cap \K_{\theta_2})_0$, which we abbreviate as $\K_{\theta_1,\theta_2}$. This gives an expression
\begin{equation*}
g = k_3 b_1 k_4 a k_5 b_2 k_6.
\end{equation*}
with $k_3,k_4,k_5,k_6 \in \K_{\theta_1,\theta_2}$, and $b_1, b_2$ in a torus maximal in $(\K_{\theta_1})^{-\theta_2}$.

\subsection{Two-qubit magic} \label{subsec:magic}
In this subsection we review some useful facts about 2-qubit gates which will also be crucial in our 3-qubit decomposition.

\begin{definition}
A \emph{magic matrix} $Q \in \U(2^n)$ is one with $Q^\dagger (\SU(2)^{\otimes n}) Q \subseteq \SO(2^n)$.
\end{definition}
For instance,
\begin{equation*}
Q = \frac{1}{2}\begin{bmatrix}
1 & 1 & i & i\\
1 & -1 & i & -i\\
-1 & 1 & i & -i\\
1 & 1 & -i & -i
\end{bmatrix}
\end{equation*}
is a magic matrix, as can be verified by checking that $Q^\dagger (X \otimes I_2) Q$ and $Q^\dagger (I_2 \otimes X) Q$ are real for $X \in \{i\sigma_x, i\sigma_y, i\sigma_z\}$.

The existence and use of magic matrices is well-known \cite{CCD, hill-wootters, makhlin}.
We take a moment to discuss a more conceptual explanation, but this material can safely be skipped.
\begin{definition}
A \emph{real structure} on $\CC^n$ is a conjugate-linear map $c : \CC^n \to \CC^n$ satisfying $c^2 = \id$.
\end{definition}
The basic example is that $c$ is complex conjugation. Viewing $c$ as a real-linear map, it has $1$- and $(-1)$-eigenspaces $V_+$ and $V_-$, which one thinks of as the ``real'' and ``imaginary'' elements of $\CC^n$ with respect to $c$. If $A$ is an $n \times n$ complex matrix commuting with $c$, then $AV_+ = V_+$, and if $B$ is an $n \times n$ matrix whose columns form an orthonormal basis of $V_+$, then $B^\dagger AB$ is a real matrix.

Let $\iota : \CC^2 \to \CC^2$ be the conjugate-linear map with $\iota(\bra{0}) = \bra{1}$ and $\iota(\bra{1}) = -\bra{0}$. Define $c = \iota^{\otimes n} : (\CC^2)^{\otimes n} \to (\CC^2)^{\otimes n}$. Thus, if $w$ is a binary word and $\tilde{w}$ denotes its negation, we have $c(\bra{w}) = (-1)^{\sum w} \bra{\tilde{w}}$. Since $\sum w + \sum \tilde{w} = n$ for any $w$, we see that $c^2 = (-1)^n$. Therefore $c$ is a real structure on $(\CC^2)^{\otimes n}$ if and only if $n$ is even. Moreover, a matrix of the form $\left[\begin{smallmatrix} a & -\overline{b}\\ b & \overline{a}\end{smallmatrix}\right]$ commutes with $\iota$, so $c$ commutes with all elements of $\SU(2)^{\otimes n}$. It follows that magic matrices exist for any even $n$: take the columns to be an orthonormal basis of $c$-fixed vectors.

One possible orthogonal basis of $V_+$ consists of the vectors
\begin{equation*}
\bra{w} + \bra{c(w)} \qquad \text{and} \qquad i(\bra{w} - \bra{c(w)})
\end{equation*}
over binary words $w$ with $w_1 = 0$. In the case $n = 2$, we make a different choice:
\begin{align*}
&\tfrac{1}{2}(\bra{00} + \bra{01} - \bra{10} + \bra{11})\\
&\tfrac{1}{2}(\bra{00} - \bra{01} + \bra{10} + \bra{11})\\
&\tfrac{1}{2}(i\bra{00} + i\bra{01} + i\bra{10} - i\bra{11})\\
&\tfrac{1}{2}(i\bra{00} - i\bra{01} - i\bra{10} - i\bra{11}),
\end{align*}
This orthonormal basis leads to the matrix \eqref{eq:magic}. The $n = 2$ case is special, in fact.
\begin{proposition} \label{prop:magic}
If $Q \in \U(4)$ is a magic matrix, then $Q^\dagger(\SU(2)^{\otimes 2})Q = \SO(4)$.
\end{proposition}
\begin{proof}
We have $Q^\dagger(\SU(2)^{\otimes 2})Q \subseteq \SO(4)$ by definition. Since both sides are connected Lie groups of dimension 6, they are equal by Lemma~\ref{lem:gp-eq}.
\end{proof}

More advanced techniques from representation theory (the \emph{Frobenius-Schur indicator}) can be used to show that magic matrices do \emph{not} exist for odd $n$. In that case the map $c$ satisfies $c^2 = -1$, defining a so-called \emph{quaternionic structure} on $(\CC^2)^{\otimes n}$ and yielding matrices $R$ such that $R^\dagger \SU(2)^{\otimes n} R$ is contained in the \emph{symplectic group} $\Sp(2^{n-1})$ (cf. Definition~\ref{def:sp}). See \cite{CCD} for more detail on these constructions.

For the rest of the paper we fix the magic matrix
\begin{equation} \label{eq:magic}
\Q = \frac{1}{2}\begin{bmatrix}
1 & 1 & i & i\\
1 & -1 & i & -i\\
-1 & 1 & i & -i\\
1 & 1 & -i & -i
\end{bmatrix}.
\end{equation}
Also set $\M = I_2 \otimes \Q$. Standard techniques for two-qubit gates show that $\Q$ is equivalent to a single CNOT up to multiplication by single-qubit gates:
\begin{equation} \label{eq:magic-circuit}
\Q = 
\begin{quantikz}
  & \gate{R_x(\pi/2)} & \ctrl{1}& \gate{R_x(-\pi)} &   &\\
  & \gate{R_z(-\pi/2)} & \targ{} & \gate{R_x(\pi/2)} & \gate{R_z(-\pi/2)} & 
\end{quantikz}
\end{equation}
Our choice of $\Q$ is somewhat arbitrary except that it seems to yield nicer formulas (e.g.\ in Figure~\ref{fig:triality} later) than some other possible choices. We have not tried to optimize $\Q$ in any rigorous sense.

The next result is a technical lemma for later use.
\begin{lemma} \label{lem:SU-I}
Let $\H = \Z(\{I_2 \otimes \sigma_y,\sigma_y \otimes \sigma_z\})$, the group of invertible matrices commuting with $I_2 \otimes \sigma_y$ and $\sigma_y \otimes \sigma_z$. Then $\Q^{\dagger} (\SU(2) \otimes I_2)\Q = \O(4) \cap \H$, which is the subgroup of matrices in $\O(4)$ of the form
\begin{equation} \label{eq:SU-I}
\begin{bmatrix}
a & -b & -c & -d\\
b & a  & d  & -c\\
c & -d & a & b\\
d & c  & -b & a
\end{bmatrix}
\end{equation}
\end{lemma}

\begin{proof}
The fact that $\O(4) \cap \H$ is the set of orthogonal matrices of the form \eqref{eq:SU-I} amounts to solving the linear system that says a matrix commutes with $I_2 \otimes \sigma_y$ and with $\sigma_y \otimes \sigma_z$. One computes
\begin{equation*}
\Q(I_2 \otimes \sigma_y)\Q^\dagger = -I_2 \otimes \sigma_y \quad \text{and} \quad \Q(\sigma_y \otimes \sigma_z)\Q^\dagger = -I_2 \otimes \sigma_x,
\end{equation*}
which commute with all elements of $\SU(2) \otimes I_2$. Thus $\Q^\dagger(\SU(2) \otimes I_2)\Q \subseteq \O(4) \cap \H$.

On the other hand, the form \eqref{eq:SU-I} shows that $\O(4) \cap \H$ is the image of $\SU(2)$ under the injective map replacing each complex entry $a+bi$ of a matrix with the $2 \times 2$ block $\left[ \begin{smallmatrix} a & -b \\ b & a \end{smallmatrix} \right]$. In particular, $\Q^\dagger(\SU(2) \otimes I_2)\Q$ and $\O(4) \cap \H$ have the same dimension and are connected, so are equal by Lemma~\ref{lem:gp-eq}.
\end{proof}

\subsection{Block matrix subgroups of $\PSO(n)$} \label{subsec:block-subgroups}
Suppose $\pi = \{\pi_1, \ldots, \pi_k\}$ is a partition of the set $[n]$. That is, $[n]$ is the disjoint union of the sets $\pi_1, \ldots, \pi_k$, called the \emph{blocks} of the partition. For instance, $\{\{1,3,4\},\{2,6\},\{5\}\}$ is a partition of $[6]$ which we will usually abbreviate as $134|26|5$ (or equivalently $26|134|5$, etc.). Let $\SO(\pi)$ be the subgroup of $\SO(n)$ consisting of block matrices whose blocks occur in positions $\pi_i \times \pi_i$ for $i = 1, \ldots, k$, and which all have determinant 1. Following our usual notation, $\PSO(\pi)$ then denotes the image of $\SO(\pi)$ in $\PSO(n)$. 
\begin{example} \hfill
\begin{itemize}
\item $\PSO(123\cdots n) = \PSO(n)$.
\item $\PSO(1|2|3|\cdots|n)$ is trivial.
\item $\P(\SO(n_1) \oplus \SO(n_2) \oplus \cdots)$ is $\PSO(\{[n_1], [n_1+1,n_1+n_2], \ldots, \})$.
\item $\PSO(13|24)$ is the group of matrices of the form $\left[ \begin{smallmatrix}
a & 0 & -b & 0\\
0 & c & 0 & -d\\
b & 0 & a & 0\\
0 & d & 0 & c
\end{smallmatrix}\right]$.
\end{itemize}
\end{example}

The \emph{meet} of two partitions $\pi$ and $\pi'$ of the same set is
\begin{equation*}
\pi \cap \pi' = \{b \cap b' : b \in \pi, b' \in \pi', b \cap b' \neq \emptyset\}.
\end{equation*}
For instance, $123|45678 \cap 234|15678 = 1|23|4|5678$. 
\begin{lemma} \label{lem:intersection} $\PSO(\pi \cap \pi') = (\PSO(\pi) \cap \PSO(\pi'))_0$.
\end{lemma}
\begin{proof}
It suffices to show that $\PSO(\pi \cap \pi')$ and $\PSO(\pi) \cap \PSO(\pi')$ have the same Lie algebra. For brevity, write $S^2$ for the set $S \times S$. The Lie algebra $\so(\pi)$ of $\PSO(\pi)$ consists of the skew-symmetric real matrices $A$  whose nonzero entries lie in the set $\bigcup_i \pi_i^2$. The Lie algebras of our two subgroups are $\so(\pi \cap \pi')$ and $\so(\pi) \cap \so(\pi')$ respectively. Indeed, the latter has its nonzero entries constrained to the set
\begin{equation*}
\bigcup_i \pi_i^2 \cap \bigcup_j (\pi_j')^2 = \bigcup_{i,j} (\pi_i \cap \pi_j')^2.
\end{equation*}
Since we can delete all empty terms $\pi_i \cap \pi_j' = \emptyset$, this by definition is $\bigcup_{k} (\pi \cap \pi')_k^2$.
\end{proof}

In general, $\PSO(\pi) \cap \PSO(\pi')$ may not be connected, e.g.\ $\PSO(123|456) \cap \PSO(14|25|36)$ consists of the 4 matrices $A \oplus A$ where
\begin{equation*}
A \in \{I_3, \Delta_{\{1\}}, \Delta_{\{2\}}, \Delta_{\{3\}}\}.
\end{equation*}
However, since $123|456 \cap 14|25|36 = 1|2|3|4|5|6$, Lemma~\ref{lem:intersection} correctly predicts that $(\PSO(123|456) \cap \PSO(14|25|36))_0$ is trivial.

If $\pi = \{\pi_-,\pi_+\}$ has just two blocks, then $\PSO(\pi)$ is a Cartan subgroup, namely $\K_{\theta}$ where $\theta = c_{\Delta_{\pi_-}}$. We can then apply the techniques of Cartan decomposition from \secsymbol\ref{subsec:cartan}, for which it is useful to have canonical parameters. We start with the slightly simpler case of $\SO(n)$.
\begin{definition}
Let $0 < p \leq n/2$. The \emph{$p$-canonical parameters} of $U \in \SO(n)$ are the singular values $\sigma(U_{[p][p]})$ of the upper-left $p \times p$ corner of $U$, together with:
\begin{itemize}
\item the sign of $\det U_{[p][p]}$ if $p < n-p$
\item the signs of $\det U_{[p][p]}$ and $\det U_{[p][n-p+1,n]}$ if $p = n/2$.
\end{itemize}
\end{definition}

\begin{proposition} \label{prop:canonical-params}
Let $\K = \SO(p) \oplus \SO(n-p)$. Then the $p$-canonical parameters of $U \in \SO(n)$ are canonical parameters with respect to the subgroup $\K$: that is, $U \in \K V \K$ if and only if $U$ and $V$ have the same $p$-canonical parameters.
\end{proposition}

\begin{example}
If $n = 2$ and $p = 1$ then $\K$ is trivial, so the $\K$-double coset of $U$ is just $\{U\}$. In this case, Proposition~\ref{prop:canonical-params} asserts that $U \in \SO(2)$ is uniquely determined by $|U_{11}|$ and the signs of $U_{11}$ and $U_{21}$.
\end{example}

\begin{proof}[Proof of Proposition~\ref{prop:canonical-params}]
We must show that $\K U \K = \K V \K$ if and only if $U$ and $V$ have the same $p$-canonical parameters. One direction is easy: left and right multiplication of $U$ by elements of $\K$ has the effect of left- and right- multiplying $U_{[p][p]}$ by a special orthogonal matrix, which changes neither its singular values nor its determinant. Note that the situation would be different if $\K$ were the larger group $\S(\O(p) \oplus \O(n-p))$, since we could change the sign of $\det U_{[p][p]}$.

Conversely, suppose $U$ and $V$ have the same $p$-canonical parameters.  By Cartan decomposition (Theorem~\ref{thm:cartan}), we may write $U = K_1 A  K_2$ and $V = K_3 B K_4$ with
\begin{equation*}
A = \begin{bmatrix} \cos D & 0 & -\sin D \\ 0 & I_{n-2p} & 0 \\ \sin D & 0 & \cos D \end{bmatrix} \quad \text{and} \quad B = \begin{bmatrix} \cos E & 0 & -\sin E \\ 0 & I_{n-2p} & 0 \\ \sin E & 0 & \cos E \end{bmatrix}
\end{equation*}
for some $p \times p$ real diagonal matrices $D,E$. Our goal is to show that $U$ and $V$ are in the same $\K$-double coset. We do this by explicitly describing how to get from $B$ to $A$ by a series of multiplications with elements of $\K$.

Since $U$ and $V$ have the same canonical parameters, $\cos D$ and $\cos E$ must be equal up to permuting diagonal entries and flipping their signs, subject to the constraint $\det \cos D = \det \cos E$. In the case $p = n/2$ we also have the assumption $\det \sin D = \det \sin E$.  By taking $P$ to be an appropriate $p \times p$ permutation matrix, we can conjugate $B$ by $P \oplus I_{n-2p} \oplus P$ so that the upper-left and lower-right blocks become $\cos D$ up to signs. Since replacing $P$ by $P \cdot \diag(-1,1,1,\ldots,1)$ in this operation has no impact on the two blocks being considered, we may assume $\det P = 1$. Thus $P \oplus I_{n-2p} \oplus P \in \K$, so this conjugation does not change the double coset of $B$. Hence we may assume that $\cos D$ and $\cos E$ are equal up to signs, in which case $\sin D$ and $\sin E$ are as well.

Next, let $q_i \in \{\pm 1\}$ be such that $q_i \cos D_{ii} = \cos E_{ii}$ for $i = 1,\ldots,p-1$. Set $q = q_1 \cdots q_{p-1}$ and $Q = \diag(q_1, \ldots, q_{p-1}, q)$. Then $Q \cos D$ and $\cos E$ are equal except perhaps for the sign of the $p$th diagonal entry. However, since $\det Q = q_1^2 \cdots q_{p-1}^2 = 1$ and $\det \cos D = \det \cos E$ by assumption, we have $\det(Q \cos D) = \cos E$, forcing $Q \cos D = \cos E$. Left-multiplying $B$ by $Q \oplus I_{n-2p} \oplus Q \in \K$ therefore allows us to assume $\cos D = \cos E$.

Similarly, if $R$ has the form $\diag(\pm 1, \cdots, \pm 1)$, then conjugating $B$ by $R \oplus I_{n-2p} \oplus R$ lets us change signs of the lower-left and upper-right blocks, leaving the other blocks unchanged. Proceeding as above, we may assume that $\sin D$ and $\sin E$ match except that perhaps $\sin D_{11} = -\sin E_{11}$. Now consider two cases.
\begin{itemize}
\item If $p = n/2$, then by assumption $\det \sin D = \det \sin E$. As in the last paragraph, this forces $\sin D = \sin E$.
\item If $p < n-p$, then the matrix $\Delta_{\{n-p,n-p+1\}}$ lies in $\K = \SO(p) \oplus \SO(n-p)$. Hence we may conjugate $B$ by it, which has the effect of flipping the sign of $\sin E_{11}$ without changing any other entries.
\end{itemize}
Either way, we can make $\sin D = \sin E$ and $\cos D = \cos E$, thereby reducing $B$ to $A$ and completing the proof.
\end{proof}

The factorization of $V \in \O(n)$ as
\begin{equation*}
V = L_1 \begin{bmatrix} \cos D & 0 & -\sin D \\ 0 & I_{n-2p} & 0 \\ \sin D & 0 & \cos  D \end{bmatrix} L_2
\end{equation*}
with $L_1,L_2 \in \O(p) \oplus \O(n-p)$ (cf. Example~\ref{ex:SO-block}) is an instance of the \emph{cosine-sine decomposition}. A numerical algorithm for computing it can be found in \cite{sutton-cosine-sine}. Using such an algorithm, we can compute matrices $K_1,K_2 \in \SO(p) \oplus \SO(n-p)$ with $K_1 V K_2 = U$ whenever $U,V \in \SO(n)$ have the same $p$-canonical parameters (also using the operations in the proof of Proposition~\ref{prop:canonical-params} to ensure that the blocks of $K_1$ and $K_2$ have determinant 1).

Strictly speaking, we are interested in the projective case $\G = \PSO(n)$ and $\K = \P(\SO(p) \oplus \SO(n-p))$, in which case a slight modification of Definition~\ref{def:canonical-params} is required to get canonical parameters. In fact the difference will never really turn out to matter, but we include the following for completeness.
\begin{corollary}
Fix $p \leq n-p$ and let the \emph{projective $p$-canonical parameters} of $U \in \PSO(n)$ be the singular values of $\sigma(U_{[p][p]})$ together with
\begin{itemize}
\item $\sgn \det U_{[p][p]}$ (if $p < n-p$ and $p$ is even)
\item $\sgn \det U_{[p][p]}$ and $\sgn \det U_{[p][n-p+1,n]}$  (if $p = n/2$ is even)
\item $\sgn \det U_{[p][p]} \cdot \sgn \det U_{[p][n-p+1,n]}$  (if $p = n/2$ is odd)
\end{itemize}
Then $U,V \in \PSO(n)$ are in the same double coset of $\K = \P(\SO(p) \oplus \SO(n-p))$ if and only if they have the same projective $p$-canonical parameters.
\end{corollary}

\begin{proof}
Two matrices $U$ and $V$ will be in the same $\K$-double coset if and only if $U$ is in the same $\SO(p) \oplus \SO(n-p)$-double coset as either $V$ or $-V$, if and only if $U$ has the same $p$-canonical parameters as either $V$ or $-V$. The last condition holds if and only if $U$ and $V$ have the same projective $p$-canonical parameters; checking this in the various cases comes down to the fact that, if $p$ is odd, then replacing $V$ by $-V$ flips the signs of $\det V_{[p][p]}$ and $\det V_{[p][n-p+1,n]}$, but not if $p$ is even.
\end{proof}

\subsection{Automorphisms and roots}
Constructing a triality symmetry of $\so(8)$ will be easier after reviewing some generalities on Lie algebras. For more detailed background, see \cite[Ch.\ III and IV]{humphreys}.

An \emph{automorphism} of a Lie group $\G$ is an invertible function $\theta : \G \to \G$ satisfying $\theta(gh) = \theta(g)\theta(h)$, and an automorphism of a Lie algebra $\g$ is an invertible linear function $\phi : \g \to \g$ satisfying $\phi([X,Y]) = [\phi(X),\phi(Y)]$. In either case, the set of automorphisms forms a group under composition, written $\Aut(\G)$ or $\Aut(\g)$ as appropriate. Fixing $g \in \G$, each conjugation function $c_g(x) = gxg^{-1}$ defines an automorphism of $\G$, and its derivative $\Ad_g$ at the identity is an automorphism of $\g$. In the case where $\G$ is a matrix Lie group, $\Ad_g : \g \to \g$ will also simply be conjugation by $g$. Automorphisms of the type $c_g$ and $\Ad_g$ are \emph{inner} automorphisms, and make up subgroups $\Inn(\G) \subseteq \Aut(\G)$ and $\Inn(\g) \subseteq \Aut(\g)$. The \emph{outer automorphism} group of $\G$ is the quotient $\Out(\G) = \Aut(\G)/\Inn(\G)$, and similarly in the Lie algebra case. A typical example of an non-inner automorphism is complex conjugation acting on $\U(n)$ or on $\mathfrak{u}(n)$. 

It follows from the classification of compact Lie algebras by their Dynkin diagrams that, with one exception, if $\g$ is the Lie algebra of a simple compact group then $\Out(\g)$ is either trivial or has a single non-identity element $\phi$ which has $\phi^2 = \id$.  For instance, take $\phi$ to be complex conjugation in the case $\g = \su(n)$ and conjugation by a reflection in the case $\g = \so(2m)$. The one exception is that $\Out(\so(8))$ has order 6, being isomorphic to the symmetric group on 3 elements. An element of order 3 in $\Out(\so(8))$ is called a \emph{triality} map, and we will construct one in \secsymbol\ref{subsec:so8-triality}.

Now suppose $\g^0$ is the Lie algebra of a compact group $\G$, and set $\g = \g^0 \otimes \CC$. Each $Z \in \g^0$ defines a linear operator $\ad_Z : \g \to \g, X \mapsto [Z,X]$, which can be shown\footnote{The formula $\exp(\ad_Z) = \Ad_{\exp(Z)}$  implies that $\exp(\RR\ad_Z)$ is compact, being the image under $\Ad$ of the compact subgroup $\exp(\RR Z) \subseteq \G$. If the Jordan form of $\ad_Z$ had a $d \times d$ Jordan block with $d > 1$, then $\exp(t\ad_Z)$ would have a matrix entry of the form $e^{\lambda t}p(t)$ with $p$ a polynomial of degree $d-1$. This entry would tend to infinity as $t \to \infty$ or $-\infty$, contradicting compactness.} to be diagonalizable. Choose a maximal abelian subalgebra $\h^0 \subseteq \g^0$ and set $\h = \h^0 \otimes \CC$. Since the operators $\{\ad_H : H \in \h\}$ commute, they are simultaneously diagonalizable. This means $\g$ breaks up as a direct sum of eigenspaces $\bigoplus_{\alpha \in \Phi \cup \{0\}} \g_{\alpha}$ satisfying 
\begin{equation*}
\ad_H(X) = [H,X] = \alpha(H)X \qquad \text{for all $H \in \h$ and $X \in \g_{\alpha}$}.
\end{equation*}
Here each eigenvalue $\alpha$ depends linearly on $H$, i.e.\ $\alpha$ is an element of $\h^*$. The zero eigenspace $\g_0$ is just $\h$ itself. The set of nonzero eigenvalues is a finite subset $\Phi \subseteq \h^*$ called the \emph{root system} of $\g$ (with respect to $\h$). Elements of $\Phi$ are \emph{roots} and their eigenspaces $\g_\alpha$ are \emph{root spaces}.

\begin{example} Say $\g^0 = \su(n)$, and take the maximal abelian subalgebra $\h^0$ to be imaginary diagonal matrices. Then $\g = \mathfrak{sl}(n,\CC)$ and $\h$ consists of all (trace zero) diagonal matrices. If $H \in \h$ is diagonal then $[H, E_{ij}] = (H_{ii}-H_{jj}) E_{ij}$, where $E_{ij}$ is the matrix with $1$ in position $(i,j)$ and $0$'s elsewhere.  Thus, if $e_i \in \h^*$ is the linear functional sending $H$ to $H_{ii}$, then $[H,E_{ij}] = (e_i-e_j)(H) E_{ij}$, so the root system is $\{e_i - e_j : i,j \in [n], i \neq j\}$,  with associated root spaces $\g_{e_i-e_j} = \CC E_{ij}$. \end{example}

Next we describe how an automorphism $\phi : \g \to \g$ is (almost) determined by its action on roots. Assume $\phi$ maps $\h$ to $\h$; it follows from \cite[Ch.\ IV, \secsymbol 16.2]{humphreys} that this assumption can always be satisfied without changing the element of $\Out(\g)$ that $\phi$ represents.  For $X \in \g_\alpha$,
\begin{equation} \label{eq:root-action}
[H, \phi(X)] = \phi [\phi^{-1}(H), X] = \alpha(\phi^{-1}(H)) \phi(X).
\end{equation}
The linear functional $H \mapsto \alpha(\phi^{-1}(H))$ is by definition $(\phi^{-1})^*(\alpha)$ where $(\phi^{-1})^*$ is the dual map to $\phi^{-1}$. In this language, \eqref{eq:root-action} shows that $\phi$ maps $\g_\alpha$ to $\g_{(\phi^{-1})^* \alpha}$. 

The final key fact is that each root space $\g_\alpha$ is 1-dimensional \cite[\secsymbol 8.4]{humphreys}.  Thus if we choose a nonzero $X_\alpha \in \g_{\alpha}$ for each $\alpha$ (a \emph{root vector}), we must have $\phi(X_{\alpha}) = c_\alpha X_{(\phi^{-1})^* \alpha}$ for some scalars $c_\alpha$.  The conclusion is that, up to these scalars,  $\phi : \g \to \g$ is completely determined by its action on $\h$.  Conversely,  if one just starts with the action of $\phi$ on $\h$, the next theorem asserts that the scalars $c_\alpha$ can be chosen so that the formula $X_{\alpha} \mapsto c_\alpha X_{(\phi^{-1})^* \alpha}$ actually defines a Lie algebra homomorphism.
\begin{theorem} \cite[Ch.\ IV, 14.2]{humphreys} \label{thm:aut-extension}
Suppose $\g$ is a complex simple Lie algebra,  and $\h$ is a maximal abelian subalgebra with corresponding root system $\Phi \subseteq \h^*$. If $f : \h \to \h$ is any invertible linear map such that $(f^{-1})^*(\Phi) \subseteq \Phi$, then there is a Lie algebra automorphism $\phi : \g \to \g$ agreeing with $f$ on $\h$.
\end{theorem}

\section{Triality} \label{sec:triality}
In this section we construct triality automorphisms for $\so(8)$ and $\PSO(8)$. One possible construction---indeed, Cartan's original construction---proceeds via the algebra of octonions \cite[\secsymbol 8.3]{quaternions-octonions}. We give a different description relying on quaternion arithmetic, which we extracted from an article of Baez \cite{baez-triality}; see also \cite[Remark 4.5]{knus-tignol-triality}.  The reader who is just interested in a quick explicit description of these maps for purposes of calculation can skip to the end of \secsymbol\ref{subsec:so8-triality} (or \secsymbol\ref{subsec:pauli}) and \secsymbol\ref{subsec:pso8-triality} respectively.

\subsection{Triality for $\so(8)$} \label{subsec:so8-triality}
As discussed in the previous section, constructing a triality automorphism mostly amounts to describing its action on roots, so we start with an explicit description of the root system and root vectors for $\so(2m)$. Let $f_{ji}$ be the matrix with $1$ in position $(j,i)$ and $-1$ in position $(i,j)$ and $0$'s elsewhere. Let $H_i = f_{2i,2i-1}$ for $i = 1, \ldots, m$ and $\h^0 = \vspan_{\RR} \{H_1, \ldots,H_m\}$. Then the ${2m \choose 2}$ matrices $\{f_{ji} : 1 \leq i < j \leq m\}$ form a basis of $\so(2m)$, and $\h^0$ is a maximal abelian subalgebra.  For instance, in the case $m = 4$,
\begin{equation*}
a_1 H_1 + a_2 H_2 + a_3 H_3 + a_4 H_4 = \left[\begin{smallmatrix}
0  & -a_1 & 0 & 0 & 0 & 0 & 0 & 0\\
a_1 & 0 & 0 & 0 & 0 & 0 & 0 & 0\\
0  & 0 & 0 & -a_2 & 0 & 0 & 0 & 0\\
0  & 0 & a_2 & 0 & 0 & 0 & 0 & 0\\
0  & 0 & 0 & 0 & 0 & -a_3 & 0 & 0\\
0  & 0 & 0 & 0 & a_3 & 0 & 0 & 0\\
0  & 0 & 0 & 0 & 0 & 0 & 0 & -a_4\\ 
0  & 0 & 0 & 0 & 0 & 0 & a_4 & 0
\end{smallmatrix}\right]
\end{equation*}
and $\h^0$ is the set of all such matrices for $a_1,a_2,a_3,a_4 \in \RR$.

Write $\so(n,\CC) := \so(n) \otimes \CC$ and $\h := \h^0 \otimes \CC$. Note that $\so(n,\CC)$ is the space of complex skew-symmetric (not skew-Hermitian!) matrices. Let $e_1, \ldots, e_m \in \h^*$ be the dual basis to $-iH_1,\ldots,-iH_m$,  so $e_j(H_k) = i\delta_{jk}$. The root system of $\so(2m,\CC)$ then consists of the $2m(m-1)$ linear functionals $\{\pm e_p \pm e_q : p \neq q\}$.  For root vectors we  take
\begin{itemize}
\item $X_{e_p-e_q} \in \so(2m,\CC)_{e_p-e_q}$ the skew-symmetric matrix with $\tfrac{1}{2}\left[ \begin{smallmatrix} 1 & i \\ -i & 1 \end{smallmatrix}\right]$ in entries $[2p-1,2p] \times [2q-1,2q]$ and $0$'s elsewhere above the diagonal,
\item $X_{e_p+e_q} \in \so(2m,\CC)_{e_p+e_q}$ the skew-symmetric matrix with $\tfrac{1}{2}\left[ \begin{smallmatrix} 1 & -i \\ -i & -1 \end{smallmatrix}\right]$ in entries $[2p-1,2p] \times [2q-1,2q]$ and $0$'s elsewhere above the diagonal,
\item $X_{-\alpha} = \overline{X_{\alpha}} \in \so(2m,\CC)_{\alpha}$ if $\alpha = e_p\pm e_q$,
\end{itemize}
where we assume $p < q$ in all cases.

\begin{example}
With $m = 4$ we have
\begin{equation*}
X_{e_1-e_2} = \frac{1}{2}\left[\begin{smallmatrix}
0 & 0 & 1 & i & 0 & 0 & 0 & 0\\
0 & 0 & -i & 1 & 0 & 0 & 0 & 0\\
-1 & i & 0 & 0 & 0 & 0 & 0 & 0\\
-i & -1 & 0 & 0 & 0 & 0 & 0 & 0\\
0 & 0 & 0 & 0 & 0 & 0 & 0 & 0\\
0 & 0 & 0 & 0 & 0 & 0 & 0 & 0\\
0 & 0 & 0 & 0 & 0 & 0 & 0 & 0\\
0 & 0 & 0 & 0 & 0 & 0 & 0 & 0
\end{smallmatrix}\right]
\end{equation*}
and the statement that $X_{e_1-e_2}$ is a root vector for $e_1-e_2$ means
\begin{equation*}
[a_1 H_1 + a_2 H_2 + a_3 H_3 + a_4 H_4, X_{e_1-e_2}] = i(a_1-a_2)X_{e_1-e_2}.
\end{equation*}
More generally, there are 24 roots for $\so(8,\CC)$ and the corresponding root vectors (unique up to scalar multiples) come from choosing one of the 6 possible $2 \times 2$ subblocks indicated in the figure 
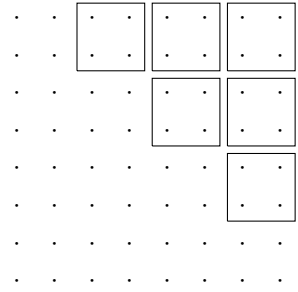
\begin{center}
\begin{tikzpicture}[scale=0.5]
\foreach \x in {0,...,7}
  \foreach \y in {0,...,7}
    \draw (\x,\y) node {$\cdot$};
\draw (1.6,5.6) -- (1.6,7.4) -- (3.4,7.4) -- (3.4,5.6) -- (1.6,5.6);
\draw (3.6,5.6) -- (3.6,7.4) -- (5.4,7.4) -- (5.4,5.6) -- (3.6,5.6);
\draw (5.6,5.6) -- (5.6,7.4) -- (7.4,7.4) -- (7.4,5.6) -- (5.6,5.6);
\draw (3.6,3.6) -- (3.6,5.4) -- (5.4,5.4) -- (5.4,3.6) -- (3.6,3.6);
\draw (5.6,3.6) -- (5.6,5.4) -- (7.4,5.4) -- (7.4,3.6) -- (5.6,3.6);
\draw (5.6,1.6) -- (5.6,3.4) -- (7.4,3.4) -- (7.4,1.6) -- (5.6,1.6);
\end{tikzpicture}
\end{center}
to fill with one of the matrices $\left[ \begin{smallmatrix} 1 & i \\ -i & 1 \end{smallmatrix}\right]$, $\left[ \begin{smallmatrix} 1 & -i \\ -i & -1 \end{smallmatrix}\right]$, or their conjugates, then applying the skew-symmetrizing operator $A \mapsto \tfrac{1}{2}(A-A^T)$.
\end{example}

For convenience, identify $\h^0 \subseteq \so(2m)$ with $\RR^m$ by taking $H_1, \ldots, H_m$ to be the standard basis, and likewise identify $i(\h^0)^* = \vspan_{\RR} \{e_1,\ldots,e_m\}$ with $\RR^m$ by taking $e_1, \ldots, e_m$ as the standard basis. For instance, $(1,0,0,-1)$ represents either $H_1-H_4$ or $e_1-e_4$ as determined by context.

Let
\begin{equation*}
D_m = \{\pm e_i \pm e_j : i,j \in [m], i \neq j\} \subseteq \RR^m
\end{equation*}
be the root system of $\so(2m,\CC)$ for $m > 1$. The \emph{root lattice} $\ZZ D_m$ is the set of integer linear combinations of elements of $D_m$, and its \emph{dual lattice} is
\begin{equation*}
(\ZZ D_m)^* = \{x \in \ZZ^m : x \cdot v \in \ZZ \text{ for all $v \in \ZZ D_m$} \}.
\end{equation*}

\begin{proposition} \label{prop:dual-ZD} For $m > 1$,
\begin{enumerate}[(a)]
\item $\ZZ D_m = \{(v_1, \ldots, v_m) \in \ZZ^m : \sum_i v_i \equiv 0 \pmod{2}\}$.
\item $(\ZZ D_m)^* = \ZZ^m \cup (\ZZ+\tfrac{1}{2})^m$.
\end{enumerate}
\end{proposition}
In words: $\ZZ D_m$ is the set of integer vectors with even sum, while $(\ZZ D_m)^*$ consists of the vectors $(x_1, \ldots, x_m)$ with the $x_i$ either being all integers or all half-integers.
\begin{proof} \hfill
\begin{enumerate}[(a)]
\item Every vector in $D_m$ has sum 0 or 2, so any integer linear combination of such vectors has even sum. Conversely, suppose $v \in \ZZ^m$ has even sum; induct on length. Since $v$ has even sum, it either has two distinct nonzero entries $v_i, v_j$, or else one entry satisfying $|v_i| \geq 2$. In the first case, set $\alpha = \sgn(v_i)e_i + \sgn(v_j)e_j \in \ZZ D_m$. In the second case, set $\alpha = \sgn(v_i)2e_i$, which is also an element of $\ZZ D_m$ since $2e_i = (e_i+e_k)+(e_i-e_k)$ where $k \neq i$ is arbitrary. In either case, $v-\alpha$ has even sum and satisfies $|v-\alpha| < |v|$, so $v-\alpha \in \ZZ D_m$ by induction. Hence $v = (v-\alpha) + \alpha \in \ZZ D_m$ as well.

\item If $x \in \ZZ^m \cup (\ZZ+\tfrac{1}{2})^m$, then $\pm x_i \pm x_j$ is an integer for all $i,j$, so $x \in (\ZZ D_m)^*$. Conversely,  suppose $x \in (\ZZ D_m)^*$, so $\pm x_i \pm x_j \in \ZZ$ for all $i,j$.  Note that this includes $i = j$, since $\ZZ D_m$ contains $2e_i$ as noted in (a). Thus $2x_i \in \ZZ$ for each $i$, i.e.\ each $x_i$ is in $\ZZ \cup (\ZZ+\tfrac{1}{2})$.  But if one particular $x_i$ is in either set $\ZZ$ or $\ZZ+\tfrac{1}{2}$, then every other $x_j$ must also be in the same set since $x_i + x_j \in \ZZ$.
\end{enumerate}
\end{proof}

We now turn to the case $m = 4$ in particular. Identify $(x_1,x_2,x_3,x_4) \in \RR^4$ with the quaternion $x_1 + x_2 i + x_3 j + x_4 k$. The only knowledge we require of the algebra of quaternions $\HH = \{x_1 + x_2 i + x_3 j + x_4 k : x_1,x_2,x_3,x_4 \in \RR\}$ is the basic multiplication rules 
\begin{equation*}
i^2 = j^2 = k^2 = -1, \qquad ij = -ji = k,\quad jk = -kj = i,\quad ki = -ik = j,
\end{equation*}
the associativity of the multiplication, and the fact that the norm
\begin{equation*}
|x_1+x_2 i+x_3 j + x_4 k| = \sqrt{x_1^2+x_2^2+x_3^2+x_4^2}
\end{equation*}
satisfies $|q_1 q_2| = |q_1||q_2|$.

The miracle that occurs after identifying $\RR^4$ with $\HH$ is that $(\ZZ D_4)^*$ is not just a lattice but a \emph{subring} of $\HH$, i.e.\ it is closed under multiplication.
\begin{definition}
The \emph{Hurwitz quaternions} are the set of quaternions of the form $x_1 + x_2 i + x_3 j + x_4 k$ where $(x_1,x_2,x_3,x_4) \in \ZZ^4 \cup (\ZZ+\tfrac{1}{2})^4$.
\end{definition}
Proposition~\ref{prop:dual-ZD}(b) identifies $(\ZZ D_4)^*$ with the set of Hurwitz quaternions. Hurwitz proposed them as a notion of ``integral'' quaternions, as they enjoy some number-theoretic benefits over the more obvious set of quaternions with integer coefficients. It is a good exercise to check that that the product of two Hurwitz quaternions must be a third.

Because of the subring property, the linear map $m_q(x) = qx$ defined by each $q \in (\ZZ D_4)^*$ maps $(\ZZ D_4)^*$ to itself. Therefore we also get a dual map $m_q^* : \ZZ D_4 \to \ZZ D_4$. In order to lead to an automorphism of $\so(8)$, this map must preserve the root system $D_4$, which consists of those integer vectors in $\ZZ D_4$ with length $\sqrt{2}$. Hence $m_q^*$ should be length-preserving, which by the property $|qx| = |q||x|$ happens exactly if $|q| = 1$. There are $24$ norm 1 Hurwitz quaternions:
\begin{equation*}
\pm 1, \pm i, \pm j, \pm k, \qquad \text{and} \qquad \frac{\pm 1 \pm i \pm j \pm k}{2}.
\end{equation*}
Set $\omega = \tfrac{-1+i+j+k}{2}$. One checks that $\omega^3 = 1$, so $(m_{\omega}^*)^3 = \id$. We have now constructed an order 3 operator $m_{\omega}^* \otimes \CC$ on $\ZZ D_4 \otimes \CC \simeq \h^*$ preserving the root system $D_4$, which extends to an order 3 automorphism of $\so(8,\CC)$ by Theorem~\ref{thm:aut-extension}. The latter is the desired triality map.

Let us make this construction more explicit. Recall that we are identifying
\begin{itemize}
\item the basis $H_1, H_2,H_3,H_4$ of $\h^0$
\item the basis $e_1,e_2,e_3,e_4$ of $i(\h^0)^*$
\item the basis $1,i,j,k$ of $\HH$.
\end{itemize}
Since
\begin{equation*}
\omega i = \frac{-1-i+j-k}{2}, \qquad \omega j = \frac{-1-i-j+k}{2}, \qquad \omega k = \frac{-1+i-j-k}{2},
\end{equation*}
the linear map $T : \h^0 \to \h^0, H \mapsto \omega H$ has matrix
\begin{equation} \label{eq:triality-matrix-h}
 \frac{1}{2} \begin{bmatrix}
-1 & -1 & -1 & -1\\
1 & -1 & -1 & 1\\
1 & 1 & -1 & -1\\
1 & -1 & 1 & -1
\end{bmatrix}.
\end{equation}
Note that this matrix is orthogonal, so under our identification of $\h^0$ and $i(\h^0)^*$, the same matrix represents the map $(T^{-1})^*$.  Accordingly, we write the action of both $T$ and $(T^{-1})^*$ as just $x \mapsto \omega x$. As per Theorem~\ref{thm:aut-extension},  we can extend $T$ to an automorphism of $\so(8,\CC)$ by sending $X_\alpha \mapsto c_\alpha X_{\omega \alpha}$ for appropriate scalars $c_\alpha$. We will not discuss a conceptual way to think about these scalars, but the required automorphism equations $\tau([X_{\alpha},X_{\beta}]) = [\tau(X_{\alpha}), \tau(X_{\beta})]$ for all $\alpha,\beta \in \Phi$ lead to a system of equations for the $c_{\alpha}$, using the fact that $[X_{\alpha}, X_{\beta}] \in \g_{\alpha+\beta}$. This yields the following explicit definition.

\begin{definition} \label{def:pso8-triality}
The \emph{triality map} $\tau : \so(8,\CC) \to \so(8,\CC)$ is the linear map with
\begin{equation*}
\tau(Y) = \begin{cases}
\omega Y & \text{if $Y \in \h$}\\
X_{\omega \alpha} & \text{if $Y = X_{\alpha}$ with $\pm \alpha \notin \{e_1-e_3, e_1+e_4, e_2+e_3,e_2+e_4\}$}\\
-X_{\omega \alpha} & \text{if $Y = X_{\alpha}$ with $\pm \alpha \in \{e_1-e_3, e_1+e_4, e_2+e_3,e_2+e_4\}$}
\end{cases}
\end{equation*}
\end{definition}
What we really want is an automorphism of the real subalgebra $\so(8) \subseteq \so(8,\CC)$. By construction, $\overline{X_\alpha} = X_{-\alpha}$ for all $\alpha$, so the real subalgebra $\so(8)$ is spanned by $\h^0$ together with all vectors $X_{\alpha}+X_{-\alpha}$ and $i(X_{\alpha}-X_{-\alpha})$. Evidently $\tau$ preserves the span of such vectors, so it restricts to an automorphism of $\so(8)$ as desired. See Example~\ref{ex:pso8-triality} below for an example calculation with $\tau$.

Even more explicitly,  consider the matrix of $\tau$ with respect to the basis $\{f_{ji} : 1 < i < j \leq 8\}$ of $\so(8)$, ordered with the pairs $(j,i)$ in lex order. Reading the entries of this $28 \times 28$ matrix left to right across rows, starting with row 1, gives the string 
\begin{center}
\texttt{\seqsplit{-66-8-4-5-6+1-4+77-4-3+2+55+6-77+2-3+4+88--5+6-5-4+55+4+7+2+67+4+3+2+77-6+1+4-55-4-88--3+6-77-2+5-4-7+2+55+6-5+4-6+66-8-4+3-6-5+4+55+4-7-2+8+6+77+2+3-4+88-+6+4-7+2-55+6+5+4+8-4+88+-3-6-77+2+1+66+8-4-5+6+1+4+77-4+3+2-78-4+3-2+77-6-1+4+8+66-8+4-}}
\end{center}
where $-$ indicates $-1/2$, $+$ indicates $1/2$, and each digit $d$ is short for a string of $d$ $0$'s (so \texttt{66} means a string of 12 zeros).

Several arbitrary choices went into our construction of $\tau$, and one could equally well call any $\alpha \circ \tau \circ \alpha^{-1}$ for $\alpha \in \Aut(\so(8))$ a triality map. Empirically, Definition~\ref{def:pso8-triality} seems to lead to nicer formulas later than some other possible choices, but we have not tried to make this idea precise.

\subsection{Triality for $\PSO(8)$} \label{subsec:pso8-triality}
We would like to lift $\tau$ to an automorphism $\T$ acting on real 3-qubit gates by setting $\T(\exp(X)) = \exp \tau(X)$. Unfortunately this simply does not work if we view $X$ as an element of $\SO(8)$. For instance,  since $\exp \left[ \begin{smallmatrix} 0 & -x \\ x & 0 \end{smallmatrix} \right] = \left[ \begin{smallmatrix} \cos x & -\sin x \\ \sin x & \cos x \end{smallmatrix} \right]$, we would have $\T(\exp(2\pi H_1)) = \T(I) = I$ but 
\begin{equation*}
\exp(\tau(2\pi H_1)) = \exp(\pi(-H_1+H_2+H_3+H_4)) = -I.
\end{equation*}
Instead we use $\PSO(8)$.

\begin{lemma} \label{lem:adjoint-aut} Let $\G$ be a connected Lie group with trivial center and Lie algebra $\g$. Any automorphism $\psi : \g \to \g$ lifts to an automorphism $\Psi : \G \to \G$ satisfying $\Psi \circ \exp = \exp \circ \psi$.
\end{lemma}

\begin{proof}
Let $\tilde{\G}$ be the simply connected universal cover of $\G$ and $q : \tilde{\G} \to \G$ the covering map. Since $q$ is a covering map, $\ker(q)$ and therefore $\Aut(\ker(q))$ are discrete groups. As $\tilde{\G}$ is connected, the homomorphism $\tilde{\G} \to \Aut(\ker(q))$ sending $g$ to $c_g$ must be constant. This shows $\ker(q)$ is a subgroup of the center $\Z(\tilde{\G})$. The assumption that $G$ has trivial center then forces $\ker(q) = \Z(\tilde{\G})$.

It is a standard fact from Lie theory that $\psi$ lifts to an automorphism $\tilde{\Psi} : \tilde\G \to \tilde\G$ with $\tilde\Psi \circ \widetilde\exp = \widetilde\exp \circ \psi$, where $\widetilde\exp$ is the exponential map for $\tilde\G$.  We want $q \circ \tilde\Psi : \tilde{\G} \to \G$ to descend to a map $\G \to \G$, which happens if and only if $\ker(q) \subseteq \ker(q \circ \tilde{\Psi})$. By the previous paragraph, this is the same as saying $\Z(\tilde{\G}) \subseteq \tilde{\Psi}^{-1}(\Z(\tilde{\G}))$.  But the right-hand side is just $\Z(\tilde{\G})$ since $\tilde{\Psi}^{-1}$ is an automorphism.
\end{proof}

\begin{definition}
The \emph{triality map} $\T : \PSO(8) \to \PSO(8)$ is the group automorphism defined by $\T(\exp(X)) = \exp(\tau(X))$.
\end{definition}
Since $\PSO(8)$ has trivial center,  Lemma~\ref{lem:adjoint-aut} assures us this is well-defined.

\begin{example} \label{ex:pso8-triality}
Let $A = \exp(\theta f_{51})$, the block matrix with $\left[ \begin{smallmatrix} \cos \theta & -\sin \theta \\ \sin \theta & \cos \theta \end{smallmatrix} \right]$ in entries $\{1,5\} \times \{1,5\}$ and the identity in the complementary block. We compute $\T(A)$.

By definition, $\T(A) = \exp(\theta \tau(f_{51}))$.  Recalling the root vectors $X_{\alpha}$ from \secsymbol\ref{subsec:so8-triality}, 
\begin{equation*}
f_{51} = -\frac{1}{2}(X_{e_1-e_3} + X_{e_1+e_3} + X_{-e_1+e_3} + X_{-e_1-e_3}).
\end{equation*}
Computing $\omega(1-j) = i+j$ and $\omega(1+j) = -1+k$ and applying Definition~\ref{def:pso8-triality},
\begin{equation*}
\tau(f_{51}) = -\frac{1}{2}(-X_{e_2+e_3} + X_{-e_1+e_4} - X_{-e_2-e_3} + X_{e_1-e_4}) = \frac{1}{2} 
\left[\begin{smallmatrix}
0 & 0 & 0 & 0 & 0 & 0 & -1 & 0\\
0 & 0 & 0 & 0 & 0 & 0 & 0 & -1\\
0 & 0 & 0 & 0 & 1 & 0 & 0 & 0\\
0 & 0 & 0 & 0 & 0 & -1 & 0 & 0\\
0 & 0 & -1 & 0 & 0 & 0 & 0 & 0\\
0 & 0 & 0 & 1 & 0 & 0 & 0 & 0\\
1 & 0 & 0 & 0 & 0 & 0 & 0 & 0\\
0 & 1 & 0 & 0 & 0 & 0 & 0 & 0
\end{smallmatrix}\right]
\end{equation*}
Using the fact that $\tau(f_{51})^2 = -1/4$ we compute
\begin{equation*}
\T(A) = \exp(\theta \tau(f_{51})) = \cos(\theta/2) + 2\sin(\theta/2) \tau(f_{51}) = 
\left[\begin{smallmatrix}
c & 0 & 0 & 0 & 0 & 0 & -s & 0\\
0 & c & 0 & 0 & 0 & 0 & 0 & -s\\
0 & 0 & c & 0 & s & 0 & 0 & 0\\
0 & 0 & 0 & c & 0 & -s & 0 & 0\\
0 & 0 & -s & 0 & c & 0 & 0 & 0\\
0 & 0 & 0 & s & 0 & c & 0 & 0\\
s & 0 & 0 & 0 & 0 & 0 & c & 0\\
0 & s & 0 & 0 & 0 & 0 & 0 & c
\end{smallmatrix} \right]
\end{equation*}
where $c = \cos(\theta/2)$ and $s = \sin(\theta/2)$.
\end{example}

Since the matrices $f_{ji}$ generate $\so(8)$, the corresponding 1-parameter families of rotations $\exp(\theta f_{ji})$ generate $\PSO(8)$. These are called \emph{Givens rotations}. One can check that $\tau(f_{ji})^2 = -1/4$ for all $i,j$, so $\exp(\theta f_{ji}) = \cos(\theta_{ji}/2) + 2\sin(\theta_{ji}/2) \tau(f_{ji})$. These facts lead to the following algorithm computing $\T(V)$ for an arbitrary $V \in \PSO(8)$, amounting to repeated applications of Example~\ref{ex:pso8-triality}.
\begin{enumerate}[(a)]
\item Write $V$ as a product of Givens rotations $\exp(\theta_{ji} f_{ji})$ in some order.
\item Compute each $\tau(f_{ji})$, e.g.\ using the fact that each $f_{ji}$ is a linear combination of the 4 root vectors $X_{\pm e_i \pm e_j}$.
\item Replace each $\exp(\theta_{ji} f_{ji})$ in the decomposition from (a) by 
\begin{equation*}
\exp(\theta_{ji} \tau(f_{ji})) = \cos(\theta_{ji}/2) + 2\sin(\theta_{ji}/2) \tau(f_{ji}).
\end{equation*}
The resulting product is $\T(V)$.
\end{enumerate}

%

\subsection{Some correspondences under triality} \label{subsec:correspondences} We now work out how $\T$ acts on some interesting subgroups $\G$ of $\PSO(8)$. Let $\mu : \SU(8) \to \SU(8)$ be conjugation by $\M^\dagger$, so $\mu(U) = \M^\dagger U \M$ and $\mu(\SU(2)^{\otimes 3}) = \SU(2) \otimes \SO(4)$. It will be convenient to record correspondences between elements $V \in \PSU(8)$ and $\T(\mu(V))$. Of course, this only makes sense if $\mu(V)$ happens to be real, so more accurately we view $\mu$ as a map $\PSO(8) \to \M\PSO(8)\M^\dagger$. The reason for the use of $\mu$ is explained in \secsymbol\ref{sec:real-circuits}.

Figure~\ref{fig:triality} shows the results of applying $\T \circ \mu$ to various one-parameter subgroups $\exp(\RR X) = \{\exp(tX) : t \in \RR\}$. Since these amount to just computing the one element $\tau(\mu(X))$, we omit the verifications. 
\begin{figure} 
\begin{equation*}  
\hspace{-1cm}\begin{array}{ccc}
\begin{array}{||c||c||}
\hline
V & \T(\mu(V))\\
\hline
\begin{tikzpicture} \node[scale=0.75] {\begin{quantikz}
 & \ctrl{2} & \gate{R_x(\alpha)} & \ctrl{2} &\\
 &          &                    &          &\\
& \targ{}  & 				   & \targ{}  &
\end{quantikz}}; \end{tikzpicture} & \raisebox{7mm}{$\exp(\alpha f_{42})$}\\
\hline
\begin{tikzpicture} \node[scale=0.75] {\begin{quantikz}
 & \ctrl{2} & 				   & \ctrl{2} &\\
 &          & \gate{R_x(\alpha)} &          &\\
& \targ{}  & 				   & \targ{}  &
\end{quantikz}}; \end{tikzpicture} & \raisebox{7mm}{$\exp(-\alpha f_{83})$}\\
\hline
\begin{tikzpicture} \node[scale=0.75] {\begin{quantikz}
 & \ctrl{2} & 				   & \ctrl{2} &\\
 &          & \gate{R_y(\alpha)} &          &\\
& \targ{}  & 				   & \targ{}  &
\end{quantikz}}; \end{tikzpicture} & \raisebox{7mm}{$\exp(\alpha f_{87})$}\\
\hline
\begin{tikzpicture} \node[scale=0.75] {\begin{quantikz}
 & \ctrl{2} & 				   & \ctrl{2} &\\
 &          & \gate{R_z(\alpha)} &          &\\
& \targ{}  & 				   & \targ{}  &
\end{quantikz}}; \end{tikzpicture} & \raisebox{7mm}{$\exp(-\alpha f_{73})$}\\
\hline
\begin{tikzpicture} \node[scale=0.75] {\begin{quantikz}
 & \ctrl{2} & 				   & \ctrl{2} &\\
 &          & 				   &          &\\
& \targ{}  & \gate{R_x(\alpha)} & \targ{}  &
\end{quantikz}}; \end{tikzpicture} & \raisebox{7mm}{$\exp(-\alpha f_{51})$}\\
\hline
\begin{tikzpicture} \node[scale=0.75] {\begin{quantikz}
 & \ctrl{2} & 				   & \ctrl{2} &\\
 &          & 				   &          &\\
& \targ{}  & \gate{R_y(\alpha)} & \targ{}  &
\end{quantikz}}; \end{tikzpicture} & \raisebox{7mm}{$\exp(-\alpha f_{65})$}\\
\hline
\begin{tikzpicture} \node[scale=0.75] {\begin{quantikz}
 & \ctrl{2} & 				   & \ctrl{2} &\\
 &          & 				   &          &\\
& \targ{}  & \gate{R_z(\alpha)} & \targ{}  &
\end{quantikz}}; \end{tikzpicture} & \raisebox{7mm}{$\exp(-\alpha f_{61})$}\\
\hline
\begin{tikzpicture} \node[scale=0.8] {\begin{quantikz}
 & \ctrl{2} & \ctrl{1} &\\
 &  	        & \targ{}  &\\
& \targ{}   & 	       &
\end{quantikz}}; \end{tikzpicture} & \raisebox{7mm}{
$\Delta_{\{2,3,6,8\}} \cdot (i \sigma_y \otimes C_1^2)$
}\\
\hline
\begin{tikzpicture} \node[scale=0.75] {\begin{quantikz}
& \gate{R_y(\alpha)}  &\\
& 				      &\\
& 				      &
\end{quantikz}}; \end{tikzpicture} & \raisebox{7mm}{$\exp(\alpha f_{64})$}\\
\hline
\end{array} & \hspace{1cm} &
\begin{array}{||c||c||}
\hline
V & \T(\mu(V))\\
\hline
\begin{tikzpicture} \node[scale=0.75] {\begin{quantikz}
& 				      &\\
& \gate{R_x(\alpha)}  &\\
& 				      &
\end{quantikz}}; \end{tikzpicture} & \raisebox{7mm}{$\exp(-\alpha f_{83})$}\\
\hline
\begin{tikzpicture} \node[scale=0.75] {\begin{quantikz}
& 				      &\\
& \gate{R_y(\alpha)}  &\\
& 				      &
\end{quantikz}}; \end{tikzpicture} & \raisebox{7mm}{$\exp(\alpha f_{87})$}\\
\hline
\begin{tikzpicture} \node[scale=0.75] {\begin{quantikz}
& 				      &\\
& \gate{R_z(\alpha)}  &\\
& 				      &
\end{quantikz}}; \end{tikzpicture} & \raisebox{7mm}{$\exp(-\alpha f_{73})$}\\
\hline
\begin{tikzpicture} \node[scale=0.75] {\begin{quantikz}
& 				      &\\
& 				      &\\
& \gate{R_x(\alpha)}  &
\end{quantikz}}; \end{tikzpicture} & \raisebox{7mm}{$\exp(-\alpha f_{51})$}\\
\hline
\begin{tikzpicture} \node[scale=0.75] {\begin{quantikz}
& 				      &\\
& 				      &\\
& \gate{R_y(\alpha)}  &
\end{quantikz}}; \end{tikzpicture} & \raisebox{7mm}{$\exp(\alpha f_{21})$}\\
\hline
\begin{tikzpicture} \node[scale=0.75] {\begin{quantikz}
& 				      &\\
& 				      &\\
& \gate{R_z(\alpha)}  &
\end{quantikz}}; \end{tikzpicture} & \raisebox{7mm}{$\exp(\alpha f_{52})$}\\
\hline
\begin{tikzpicture} \node[scale=0.75] {\begin{quantikz}
 & \ctrl{1} & \gate{R_x(\alpha)} & \ctrl{1} &\\
 &  \targ{} &                    & \targ{}  &\\
& 		  & 				   &          &
\end{quantikz}}; \end{tikzpicture} & \raisebox{7mm}{$\exp(-\alpha f_{76})$}\\
\hline
\begin{tikzpicture} \node[scale=0.75] {\begin{quantikz}
 & \ctrl{1} & 				   & \ctrl{1} &\\
 &  \targ{} & \gate{R_x(\alpha)} & \targ{}  &\\
& 		  & 				   &          &
\end{quantikz}}; \end{tikzpicture} & \raisebox{7mm}{$\exp(-\alpha f_{83})$}\\
\hline
\begin{tikzpicture} \node[scale=0.75] {\begin{quantikz}
 & \ctrl{1} & 				   & \ctrl{1} &\\
 &  \targ{} & \gate{R_y(\alpha)} & \targ{}  &\\
& 		  & 				   &          &
\end{quantikz}}; \end{tikzpicture} & \raisebox{7mm}{$\exp(\alpha f_{43})$}\\
\hline
\begin{tikzpicture} \node[scale=0.75] {\begin{quantikz}
 & \ctrl{1} & 				   & \ctrl{1} &\\
 &  \targ{} & \gate{R_z(\alpha)} & \targ{}  &\\
& 		  & 				   &          &
\end{quantikz}}; \end{tikzpicture} & \raisebox{7mm}{$\exp(\alpha f_{84})$}\\
\hline
\end{array}
\end{array}
\end{equation*}
\caption{Triality correspondences} \label{fig:triality} 
\end{figure}

The entries in Figure~\ref{fig:triality} can be used to work out what happens to larger subgroups. For instance:
\begin{proposition} \label{prop:triality-correspondences}
\hfill
\begin{enumerate}[(a)]
\item $(\T \circ \mu)\P(I_4 \otimes \SU(2)) = \PSO(125|3|4|6|7|8) \subseteq \PSO(8)$.
\item $(\T \circ \mu)\P(I_2 \otimes \SU(2) \otimes I_2) = \PSO(378|1|2|4|5|6)$.
\item $(\T \circ \mu)\P(C_1^2(I_2 \otimes \SU(2) \otimes I_2)C_1^2) = \PSO(348|1|2|5|6|7)$.
\end{enumerate}
\end{proposition}

\begin{proof}
For part (a), since $I_4 \otimes \SU(2)$ is the group generated by the one-parameter families $I_4 \otimes R_a(\RR)$ for $a \in \{x,y,z\}$, Figure~\ref{fig:triality} shows that $(\T \circ \mu)\P(I_4 \otimes \SU(2))$ is the group generated by the one-parameter families $\exp(\RR f_{51})$, $\exp(\RR f_{21})$, and $\exp(\RR f_{52})$. Identifying $\PSO(125|3|4|6|7|8)$ with $\PSO(3)$, these are just the families of rotations around the $y$-, $z$-, and $x$-axes respectively, which together generate $\PSO(3)$. Parts (b) and (c) are analogous.
\end{proof}

An interesting feature of the triality map is that it realizes all of the exceptional isomorphisms between the classical matrix Lie algebras of types $A_n, B_n, C_n, D_n$, in various ways. In this way, it is similar to the magic basis change $\mu$, which realizes the isomorphism $\SO(4) \simeq (\SU(2) \times \SU(2))/\{\pm 1\}$. For instance:
\begin{itemize}
\item $A_1 \simeq B_1$: Proposition~\ref{prop:triality-correspondences}(a-c) show $\PSU(2) \simeq \PSO(3)$.
\item $D_2 \simeq A_1 \oplus A_1$: One can check $\T^{-1}\PSO(125|378|4|6) = \P(I_2 \otimes \SO(4))$, so $\PSO(4) \simeq \PSO(3) \times \PSO(3)$.
\item $A_3 \simeq D_3$: One can check that $\T^{-1}\PSO(123578|4|6)$ is the group
\begin{equation*}
\left\{ \left[\begin{smallmatrix} \Re(U) & -\Im(U) \\ \Im(U) & \Re(U) \end{smallmatrix} \right] : U \in \SU(4) \right\},
\end{equation*}
so $\SU(4) / \{\pm I\} \simeq \PSO(6)$.
\item $C_2 \simeq B_2$: $\P(\Sp(2)) \simeq \PSO(5)$ follows from Lemma~\ref{lem:triality-sp2} in the next section, where the symplectic groups $\Sp(n)$ are also defined.
\end{itemize}

\subsection{Pauli operators and triality} \label{subsec:pauli} 
In \secsymbol\ref{subsec:so8-triality} we used the basis of skew-symmetric unit matrices $\{f_{ji} : 1 \leq i < j \leq 8\}$ of $\so(8)$. It turns out that, up to scalars, the triality map identifies these matrices with another natural basis of $\so(8)$, the basis of \emph{Pauli operators}.

We will use a word $abc\cdots$ on the alphabet $\{X,Y,Z,I\}$ as an abbreviation for the tensor product $i\sigma_a \otimes \sigma_b \otimes \sigma_c \otimes \cdots$, where $\sigma_I$ means the $2 \times 2$ identity $I_2$. For instance,
\begin{equation*}
YXI = i\sigma_y \otimes \sigma_x \otimes I_2 = \left[ 
\begin{smallmatrix}
0 & 0 & 0 & 0 & 0 & 0 & 1 & 0\\
0 & 0 & 0 & 0 & 0 & 0 & 0 & 1\\
0 & 0 & 0 & 0 & 1 & 0 & 0 & 0\\
0 & 0 & 0 & 0 & 0 & 1 & 0 & 0\\
0 & 0 & -1 & 0 & 0 & 0 & 0 & 0\\
0 & 0 & 0 & -1 & 0 & 0 & 0 & 0\\
-1 & 0 & 0 & 0 & 0 & 0 & 0 & 0\\
0 & -1 & 0 & 0 & 0 & 0 & 0 & 0
\end{smallmatrix} \right] \end{equation*}
Each $\sigma_a$ is Hermitian, so such a tensor product is skew-Hermitian. Indeed, the $4^n$ $n$-qubit Pauli operators $W \in \{I,X,Y,Z\}^n$ form a basis of $\su(2^n)$. Since $\sigma_x, \sigma_z, I_2$ are real while $\sigma_y$ is pure imaginary, a Pauli operator $W$ is real iff it has an odd number of $Y$'s. One can calculate using the binomial theorem that the number of such operators is ${2^n \choose 2} = \dim \so(2^n)$, so they form a basis of $\so(2^n)$.

The following array shows the correspondence between Pauli operators and skew-symmetric unit matrices $f_{ji}$, where a Pauli operator $\pm W$ in position $(j,i)$ means that $\tau(W) = \pm 2f_{ji}$:
\begin{equation} \label{eq:pauli-correspondence}
\text{\texttt{
\begin{tabular}{cccccccc}
· & · & · & · & · & · & · & · \\                          
+IIY & · & · & · & · & · & · & ·\\ 
+YXZ & -YXX & · & · & · & · & · & ·\\
+XYX & +XYZ & -ZZY & · & · & · & · & ·\\
-IYZ & +IYX & +YZI & -XIY & · & · & · & ·\\
-ZYX & -ZYZ & -XZY & -YII & -ZIY & · & · & ·\\
-YZZ & +YZX & -IYI & +ZXY & -YXI & +XXY & · & ·\\
-YIX & -YIZ & +IXY & +ZYI & -YYY & +XYI & -IZY & ·
\end{tabular}}}
\end{equation}
For instance, $\tau(i\sigma_y \otimes \sigma_x \otimes \sigma_x) = -2f_{32}$. Alternatively, one could view a Pauli operator $W$ as an element of $\PSO(8)$, in which case \eqref{eq:pauli-correspondence} says that
\begin{equation*}
\T(W) = \T(\exp(\pi W/2)) = \exp(\pm \pi f_{ji}) = J_{\{i,j\}}
\end{equation*}
if $W$ occurs in entry $(j,i)$.

\begin{corollary}  \label{cor:pauli}
Suppose $c : \so(8) \to \so(8)$ is conjugation by some diagonal orthogonal matrix $\Delta_S$, and set $c' = \tau^{-1} \circ c \circ \tau$. Then the Pauli basis is an eigenbasis for the involution $c'$. 
\end{corollary}

\begin{proof}
The matrices $2f_{ji}$ form an eigenbasis of $c$, with eigenvalue $-1$ if $|\{i,j\} \cap S| = 1$ and $+1$ otherwise. Therefore the matrices $\tau^{-1}(2f_{ji})$ form an eigenbasis of $\tau^{-1} \circ c \circ \tau$.
\end{proof}

The significance of this is that it makes it easy to find the Cartan decomposition $\k \oplus \p$ with respect to an involution of the form $c' = \tau^{-1} \circ c \circ \tau$: $\k$ is the span of the Pauli operators fixed by $c'$, while $\p$ is spanned by those negated by $c'$.

\section{Real circuit elements} \label{sec:real-circuits}
It is \emph{a priori} unclear when a long product of CNOTs and single-qubit gates actually multiplies to a real matrix. This is one advantage of using the magic matrix $\M$: we know at least that any element of $\M^\dagger (R_y(\RR) \otimes \SU(2) \otimes \SU(2))\M$ is guaranteed to be real. In this section we discuss some other basic circuit elements that will be useful in constructing real 3-qubit gates. These considerations also lead us toward to an interesting subgroup of $\PSO(8)$ that turns out to be key to our Cartan decomposition in \secsymbol\ref{sec:pso8-cartan} later.

\begin{lemma} \label{lem:real-gates-rx}
Let $C \in \{C_1^2, C_1^3\}$ and let $\G$ be the subgroup
\begin{equation*}
C(R_x(\RR) \otimes \SU(2) \otimes \SU(2))C \subseteq \SU(8)
\end{equation*}
Then $\mu(\G) = \M^\dagger \G \M \subseteq \SO(8)$.
\end{lemma}

\begin{proof}
One checks that if $C$ has control qubit 1 and target qubit 2 or 3, then $\mu(C)$ has the form $\diag(I_4, iS)$ with $S$ a real matrix. First suppose $U \in I_2 \otimes \SU(2) \otimes \SU(2)$. Then $\mu(U) \in I_2 \otimes \SO(4)$ has the form $W \oplus W$ for $W$ real, so
\begin{equation*}
\mu(CUC) = (I_4 \oplus iS)(W \oplus W)(I_4 \oplus iS) = W \oplus -SWS 
\end{equation*}
is real.

Next, suppose $U = R_x(\theta) \otimes I_4$, so
\begin{equation*}
\M^\dagger C U C \M = \exp(-i\M^\dagger C (\sigma_x \otimes I_4) C \M \cdot \theta/2).
\end{equation*}
Since $\M = I_2 \otimes \Q$ commutes with $\sigma_x \otimes I_4 = \left[ \begin{smallmatrix} 0 & I_4 \\ I_4 & 0 \end{smallmatrix} \right]$, the argument to the matrix exponential above is
\begin{align*}
-i\M^\dagger C(\sigma_x \otimes I_4) C \M &= -i\M^\dagger C\M\cdot (\sigma_x \otimes I_4)\cdot \M^\dagger C \M\\
&= -i(I_4 \oplus iS) \begin{bmatrix} 0 & I_4 \\ I_4 & 0 \end{bmatrix} (I_4 \oplus iS) = \begin{bmatrix} 0 & S \\ S & 0 \end{bmatrix},
\end{align*}
again a real matrix.
\end{proof}

\begin{definition} \label{def:sp}
The \emph{symplectic group} $\Sp(n) \subseteq \U(2n)$ is the subgroup of matrices $U$ where each $2 \times 2$ block $U_{[2i-1,2i][2i-1,2i]}$ has the form $\left[ \begin{smallmatrix} w & -\bar{z} \\ z & \bar{w} \end{smallmatrix} \right]$.
\end{definition}
For instance, $\Sp(1) = \SU(2)$. The symplectic group $\Sp(n)$ is a compact simply connected group of dimension $n(2n+1)$. Setting $\Omega = I_n \otimes i\sigma_y = \left[ \begin{smallmatrix} 0 & 1 \\ -1 & 0 \end{smallmatrix} \right]^{\oplus n}$, an alternative phrasing of Definition~\ref{def:sp} is that $\Sp(n) = \{U \in \U(2n) : \Omega \overline{U} \Omega^T = U\}$. This reveals $\Sp(n)$ as a Cartan subgroup of $\U(2n)$. 

As an aside, the algebra of quaternions $\HH$ is isomorphic to the algebra of $2 \times 2$ matrices of the form $\left[ \begin{smallmatrix} w & -\bar{z} \\ z & \bar{w} \end{smallmatrix} \right]$ by identifying this matrix with $w+jz$. Making this replacement for each $2 \times 2$ block reveals $\Sp(n)$ as the group of $n \times n$ unitary quaternionic matrices. We will not use this perspective, but it has the nice feature of unifying the classical compact matrix groups $\O(n)$, $\U(n)$, and $\Sp(n)$ as being the unitary matrices whose entries come from the three possible associative real division algebras $\RR$, $\CC$, and $\HH$.

\begin{lemma} \label{lem:real-gates-sp} The subgroup of $\U(4)$ generated by
\begin{equation*}
\H_1 = R_y(\RR) \otimes \SU(2) \quad \text{and} \quad \H_2 = C_1^2 (R_x(\RR) \otimes \SU(2))C_1^2
\end{equation*}
is
\begin{equation*}
\Sp(2) = \left\{\left[\begin{smallmatrix} a & -\overline{b} & c & -\overline{d} \\
								   b & \overline{a} & d & \overline{c}\\
								   e & -\overline{f} & g & -\overline{h}\\
							       f & \overline{e} & h & \overline{g}
\end{smallmatrix}\right] \in \U(4) \right\}
\end{equation*}
\end{lemma}

\begin{proof}
It suffices to show that the Lie algebras
\begin{align*}
\h_1 &= i \RR\sigma_y \otimes I_2 + I_2 \otimes \su(2) = \left\{ \left[ \begin{smallmatrix}
ir & -\overline{a} & -s & 0\\
a &  -ir & 0 & -s\\
s & 0 & ir & -\overline{a}\\
0 & s & a & -ir \end{smallmatrix} \right] : r,s \in \RR\right\} \\
\h_2 &= C_1^2 (i \RR \sigma_x \otimes I_2 + I_2 \otimes \su(2)) C_1^2 = \left\{ \left[ \begin{smallmatrix}
ir & -\overline{a} & 0 & is\\
a &  -ir & is & 0\\
0 & is & -ir & a\\
is & 0 & -\overline{a} & ir
\end{smallmatrix} \right] : r,s \in \RR\right\}
\end{align*}
generate the Lie algebra
\begin{equation*}
\sp(2) = \left\{ \left[ \begin{smallmatrix}
ir & -\overline{a} & -b & -\overline{c}\\
a &  -ir 		  & c   & -\overline{b}\\
b & -\overline{c}  & is & -\overline{d}\\
c & \overline{b}   & d  & -is
\end{smallmatrix}\right] : r,s \in \RR \right\}.
\end{equation*}

Let $\theta : \Sp(2) \to \Sp(2)$ be conjugation by $\sigma_x \otimes \sigma_x$, so $\theta$ (and $d\theta$) have the effect of rotating a matrix by $180^\circ$. The subspace $\k \subseteq \sp(2)$ fixed by $d\theta$ is exactly $\h_2$. 

The $(-1)$-eigenspace of $d\theta$ is
\begin{equation} \label{eq:sp2-cartan}
\p = \left\{ \left[ \begin{smallmatrix}
ir & -\overline{a} & -\overline{b} & -s\\
a &  -ir 		  & s   & -b\\
b & -s  & ir & -a\\
s & \overline{b}   & \overline{a}  & -ir
\end{smallmatrix}\right] : r,s \in \RR \right\}.
\end{equation}
Let $\a$ be the span of the two commuting elements
\begin{equation*}
i\sigma_y \otimes I_2 = \left[ \begin{smallmatrix}
0 & 0 & 1 & 0\\
0 & 0 & 0 & 1\\
-1 &0 & 0 & 0\\
0 & -1 & 0 & 0
\end{smallmatrix} \right] \qquad \text{and} \qquad I_2 \otimes i\sigma_y = \left[ \begin{smallmatrix}
0 & 1 & 0 & 0\\
-1 & 0 & 0 & 0\\
0 & 0 & 0 & 1\\
0 & 0 & -1 & 0
\end{smallmatrix} \right].
\end{equation*}
Note that $\a \subseteq \h_1$ and that $\a$ is also contained in the $(-1)$-eigenspace $\p$ of $d\theta$. We claim $\a$ is a maximal abelian subalgebra of $\p$. Indeed, a matrix $A$ commutes with $I_2 \otimes i\sigma_y$ if and only if each $2 \times 2$ block has the form $\left[ \begin{smallmatrix} x & -y \\ y & x \end{smallmatrix} \right]$; in particular, if $A$ has the form \eqref{eq:sp2-cartan} then we must have $r = 0$ and $a,b \in \RR$. Similarly, $A$ commutes with $i\sigma_y \otimes I_2$ iff $a,b \in \RR$ and $s = 0$. Thus if both conditions hold then $A =  \left[ \begin{smallmatrix}
0 & -a & -b & 0\\
a &  0 		  & 0   & -b\\
b & 0  & 0 & -a\\
0 & b   & a  & 0
\end{smallmatrix}\right]$, which already lies in $\a$. Cartan decomposition (specifically Corollary~\ref{cor:ka}) therefore shows $\sp(2)$ is generated by $\k = \h_2$ together with $\a \subseteq \h_1$.
\end{proof}

Combining Lemmas~\ref{lem:real-gates-rx} and \ref{lem:real-gates-sp} gives:
\begin{corollary} \label{cor:real-sp}
$\mu\P(\Sp(2) \otimes \SU(2))\subseteq \SO(8)$.
\end{corollary}

The discussion above identifies $\mu\P(\Sp(2) \otimes \SU(2))$ as a subgroup of $\PSO(8)$ determined by relatively simple circuit elements. The second key fact about it is that it is a Cartan subgroup, which we leverage in the next section to derive our new circuit for $\PSO(8)$.
\begin{lemma} \label{lem:triality-sp2} $(\T \circ \mu)\P(\Sp(2) \otimes I_2) = \PSO(34678|1|2|5)$.
\end{lemma}
\begin{proof}
Set $\Omega = I_2 \otimes i\sigma_y \otimes I_2$ and $\H = \mu\P(\Sp(2) \otimes \SU(2))$, so $\Omega U \Omega^T = \overline{U}$ for $U \in \Sp(2) \otimes I_2$. Making the substitution $U = \M V \M^\dagger$ for $V \in \H \subseteq \SO(8)$ and rearranging shows that $V$ commutes with $\M^T \Omega \M = I_4 \otimes i\sigma_y$. Evidently any $V \in \H$ also commutes with any element of $\mu\P(I_4 \otimes \SU(2))$.

We see from the above that $\T(\H)$ commutes with both
\begin{equation*}
\T(\M^T \Omega \M) = \Delta_{\{1,2\}} \qquad \text{and} \qquad \T(\mu(I_4 \otimes i\sigma_z)) = \Delta_{\{2,5\}}.
\end{equation*}
Since $\T(\H)$ is connected, 
\begin{align} \label{eq:secret-PSO5}
\T(\H) &\subseteq \Z(\{\Delta_{\{1,2\}},\Delta_{\{2,5\}}\})_0 = (\PSO(12|345678) \cap \PSO(134678|25))_0 \nonumber\\
&= \PSO(34678|1|2|5),
\end{align}
using Lemma~\ref{lem:intersection}. But both $\T(\H) \simeq \Sp(2)$ and $\PSO(5)$ have dimension 10, so the inclusion \eqref{eq:secret-PSO5} is an equality by Lemma~\ref{lem:gp-eq}.
\end{proof}

\begin{theorem} \label{thm:key-subgroup} Let $\psi : \PSO(8) \to \PSO(8)$ be conjugation by $\Delta_{\{1,2,5\}}$, and let $\chi = \T^{-1} \circ \psi \circ \T$. Then $\K_{\chi} = \mu \P(\Sp(2) \otimes \SU(2)) \subseteq \PSO(8)$.
\end{theorem}

\begin{proof}
Since $V \in \PSO(8)$ is fixed by $\psi$ if and only if $\T^{-1}(V)$ is fixed by $\chi$, we have $\K_{\chi} = \T^{-1}(\K_{\psi})$. Now
\begin{align*}
\K_{\chi} &= \T^{-1}(\K_{\psi}) = \T^{-1}(\PSO(125|34678))\\
&= \T^{-1}(\PSO(125|3|4|6|7|8) \cdot \PSO(34678|1|2|5))\\
&= \mu\P(I_4 \otimes \SU(2)) \cdot \mu\P(\Sp(2) \otimes I_4),
\end{align*}
using Proposition~\ref{prop:triality-correspondences}(a) and Lemma~\ref{lem:triality-sp2} in the last line.
\end{proof}

Given the utility of the slightly mysterious subgroup $\K_\chi = \mu \P(\Sp(2) \otimes \SU(2))$, it would be useful to have other descriptions; here is one possibility.
\begin{proposition} \label{prop:K1-2}
$\mu \P(\Sp(2) \otimes \SU(2))$ is the set of matrices $V = \left[\begin{smallmatrix} V_{11} & V_{12} \\ V_{21} & V_{22} \end{smallmatrix}\right]$ where the $4 \times 4$ blocks $V_{ij}$ have the property that each $V_{ij}V_{kl}^T$ commutes with $I_2 \otimes \sigma_y$ and $\sigma_y \otimes \sigma_z$, i.e.\ has the form
\begin{equation} \label{eq:SU-I-2}
\begin{bmatrix}
a & -b & -c & -d\\
b & a  & d  & -c\\
c & -d & a & b\\
d & c  & -b & a
\end{bmatrix}.
\end{equation}
\end{proposition}

\begin{proof} 
Suppose that $V = \mu(S \otimes W)$ where $S \in \Sp(2)$ and $W \in \SU(2)$. By definition, $S = \left[ \begin{smallmatrix} S_{11} & S_{12} \\ S_{21} & S_{22} \end{smallmatrix} \right]$ where each $S_{ij}$ has the form $\left[ \begin{smallmatrix} a & -\overline{b} \\ b & \overline{a} \end{smallmatrix}\right]$. Dividing by $|a|^2+|b|^2$ shows that each $S_{ij}$ can be written $c_{ij}U_{ij}$ where $c_{ij} \in \RR$ and $U_{ij} \in \SU(2)$. Since $\M = I_2 \otimes \Q = \diag(\Q,\Q)$,
\begin{align} \label{eq:K1-2}
V &= \M^{-1}(S \otimes W)\M = \begin{bmatrix} \Q & 0 \\ 0 & \Q \end{bmatrix}^\dagger \begin{bmatrix} S_{11} \otimes W & S_{12} \otimes W \\ S_{21} \otimes W & S_{22} \otimes W \end{bmatrix} \begin{bmatrix} \Q & 0 \\ 0 & \Q \end{bmatrix} \nonumber\\
&= \begin{bmatrix} c_{11}\Q^\dagger(U_{11} \otimes W) \Q & c_{12}\Q^\dagger(U_{12} \otimes W) \Q \\ c_{21}\Q^\dagger(U_{21} \otimes W) \Q & c_{22}\Q^\dagger(U_{22} \otimes W) \Q \end{bmatrix}
\end{align}
This calculation shows that each $V_{ij}V_{kl}^T$ lies in the subgroup $\Q^\dagger (\SU(2) \otimes I_2) \Q$ up to a scalar multiple. By Lemma~\ref{lem:SU-I}, this is equivalent to the stated condition.

Conversely, suppose the stated conditions on $V$ hold. Cartan decomposition with respect to $\S(\O(4) \times \O(4)) \subseteq \SO(8)$ lets us write
\begin{equation*}
V = \begin{bmatrix} O_1 & 0 \\ 0 & O_2 \end{bmatrix} \begin{bmatrix} C & -S \\ S & C \end{bmatrix} \begin{bmatrix} O_3 & 0 \\ 0 & O_4 \end{bmatrix} = \begin{bmatrix} O_1 C O_3 & -O_1 S O_4 \\ O_2 C O_3 & O_2 S O_4 \end{bmatrix}
\end{equation*}
where $C$ and $S$ are nonnegative real diagonal with $C^2+S^2=I_4$ and $O_1, \ldots, O_4 \in \O(4)$. The matrix $V_{11}V_{11}^T = O_1 C^2 O_1^T$ is symmetric and also has the form \eqref{eq:SU-I-2} by assumption, which together force it to be a scalar matrix. This shows $C^2$ and hence $C$ are scalar matrices. The same argument with other blocks shows that $S$ is a scalar matrix.

Thus, each block $V_{ij}$ actually has the form $c_{ij}O_{ij}$ for $c_{ij} \in \RR$ and $O_{ij} \in \O(4)$. Moreover, each $V_{ij}V_{kl}^T = c_{ij}c_{kl} O_{ij}O_{kl}^T$ lies in the subgroup $\Q(\SU(2) \otimes I_2)\Q^\dagger$ up to a scalar multiple by assumption and by Lemma~\ref{lem:SU-I}. Reversing the argument in the first paragraph of the proof, we conclude that $V \in \mu \P(\Sp(2) \otimes \SU(2))$.
\end{proof}

We can also describe the Lie algebra of $\K_\chi = \mu \P(\Sp(2) \otimes \SU(2))$.
\begin{proposition}
The Lie subalgebra of $\so(8)$ corresponding to $\K_\chi$ is
\begin{enumerate}[(a)]
\item the span of the Pauli operators
\begin{equation*}
\{XXY,XYI,XZY,ZXY,ZYI,ZZY,YII,IXY,IYX,IYZ,IYI,IZY,IIY\};
\end{equation*}
\item the subspace of matrices $A \in \so(8)$ with $\tr(AB) = 0$ for the 15 Pauli operators $B$ not listed in part (a);
\item the space of real skew-symmetric matrices $H = \left[ \begin{smallmatrix} H_{11} & -H_{21}^T \\ H_{21} & H_{22} \end{smallmatrix} \right]$ where the $4 \times 4$ blocks $H_{ij}$ have the property that each $H_{ij} - H_{kl}$ commutes with $I_2 \otimes \sigma_y$ and $\sigma_y \otimes \sigma_z$, i.e.\ has the form
\begin{equation*}
\begin{bmatrix}
0 & -r & -s & -t\\ 
r & 0  & t & -s\\
s & -t & 0 & r\\
t & s  & -r & 0
\end{bmatrix}.
\end{equation*}
\end{enumerate}
\end{proposition}

\begin{proof} \hfill
\begin{enumerate}[(a)]
\item By Theorem~\ref{thm:key-subgroup}, $\K_\chi = \mu\P(\Sp(2) \otimes \SU(2))$ is the fixed-point subgroup of an involution of the type considered in Corollary~\ref{cor:pauli}. By that corollary, its Lie algebra is spanned by the Pauli operators fixed by $d\chi$. These can be easily read off of \eqref{eq:pauli-correspondence}: they are the elements in positions $\{1,2,5\}^2 \cup \{3,4,6,7,8\}^2$, since those are the pairs $(j,i)$ such that $f_{ji}$ is fixed by $\psi$.

\item The Pauli operators form an orthogonal basis of $\so(n)$ under the inner product $(A,B) = \tr(AB)$: this follows from the fact that $\tr(AB) = 0$ if $A,B \in \{I_2,\sigma_x,\sigma_y,\sigma_z\}$ are distinct, plus the general property $\tr(P \otimes Q) = \tr(P)\tr(Q)$. Therefore (a) implies (b).

\item Follows by taking derivatives in Proposition~\ref{prop:K1-2}.
\end{enumerate}
\end{proof}

One might wonder whether this construction can be simplified. One issue is that, unlike more familiar involutions like complex conjugation or conjugation by a fixed matrix, $\chi$ is not a linear function. That is, the entries of $\chi(V)$ do not depend linearly on the entries of $V$. For instance,
\begin{equation*}
\chi\left( \left[ \begin{smallmatrix} \cos t & -\sin t \\ \sin t & \cos t \end{smallmatrix} \right] \oplus I_6 \right) = \left[\begin{smallmatrix} c & s \\ -s & c \end{smallmatrix} \right] \oplus \left[\begin{smallmatrix} c & -s \\ s & c \end{smallmatrix} \right] \oplus \left[\begin{smallmatrix} c & -s \\ s & c \end{smallmatrix} \right] \oplus \left[\begin{smallmatrix} c & -s \\ s & c \end{smallmatrix} \right]
\end{equation*}
where $c = \cos(t/2)$ and $s = \sin(t/2)$. 

Perhaps $\K_{\chi}$ can be described as the fixed-point set of an entirely different involution $\phi$? If this were the case, then since $d\phi$ is an automorphism, it would be an isometry with respect to the Killing form on $\so(8)$. Its $(-1)$-eigenspace is therefore the orthogonal complement $\k_{\chi}^\perp$ of the $1$-eigenspace $\k_\chi$. By the same logic applied to $d\chi$, $\k_{\chi}^\perp$ is simply $\p$. Therefore $d\phi$ and $d\chi$ have the same eigenvalues and eigenspaces, i.e.\ they are the same linear map. Since $d\phi$ determines $\phi$ by the formula $\phi(\exp(X)) = \exp(d\phi(X))$, this means $\phi = \chi$.

\section{Cartan decompositions for $\PSO(8)$} \label{sec:pso8-cartan}
As with many other circuit decompositions in quantum computing \cite{CCD,drury-love,khaneja-glaser,block-ZXZ,quantum-shannon,wei-di}, ours relies on iterated Cartan decomposition. Here is a more specific outline, which the methods of \cite{wei-di} also follow.
\begin{enumerate}[1.]
\item Choose commuting involutive automorphisms $\theta_1, \theta_2$ of $\PSO(8)$, and sets
\begin{itemize}
\item $\A \subseteq \PSO(8)$ with $\K_{\theta_1}\A\K_{\theta_1} = \PSO(8)$
\item $\B \subseteq \PSO(8)$ with $\K_{\theta_1,\theta_2}\B\K_{\theta_1,\theta_2} = \K_{\theta_1}$.
\end{itemize}
\item Given $V \in \PSO(8)$, use two Cartan decompositions as in \secsymbol\ref{subsec:commuting-cartan} to write
\begin{equation} \label{eq:double-decomp}
V = K_1 B_1 K_2 A K_3 B_2 K_4
\end{equation}
where $A \in \A$ and $B_1,B_2 \in \B$ and $K_1,\ldots,K_4 \in \K_{\theta_1,\theta_2}$. 

\item Find explicit decompositions for $\M$ and for elements of $\mu^{-1}(\A)$ and $\mu^{-1}(\B)$ and $\mu^{-1}(\K_{\theta_1,\theta_2})$ in terms of CNOTs and single-qubit gates.
\end{enumerate}

This works because we can rewrite \eqref{eq:double-decomp} as
\begin{align*}
V &= \M^\dagger \mu^{-1}(V) \M = \M^\dagger \mu^{-1}(K_1 B_1 K_2 A K_3 B_2 K_4) \M\\
&= \M^\dagger \mu^{-1}(K_1)\mu^{-1}(B_1)\cdots \mu^{-1}(K_4)\M.
\end{align*}
As described in more detail in \secsymbol \ref{sec:real-circuits}, the point of applying $\mu^{-1}$ and then its inverse here is to make it easier to see which circuit elements will lead to a real matrix.

There is a sense in which $\PSO(8)$ is a rich setting for Cartan decomposition techniques: among all connected simple compact Lie groups, only $\PSO(8)$ and its simply connected cover $\operatorname{Spin}(8)$ have more than one distinct outer automorphism of order 2.

\begin{example} \label{ex:wei-di}
Take $\theta_1$ and $\theta_2$ to be conjugation by $\Delta_{[4]}$ and by $i\sigma_y \otimes I_4 = \left[ \begin{smallmatrix} 0 & I_4 \\ -I_4 & 0 \end{smallmatrix} \right]$, respectively. Then
\begin{equation*}
\K_{\theta_1} = \P(\SO(4) \oplus \SO(4)) \quad \text{and} \quad \K_{\theta_1,\theta_2} = I_2 \otimes \PSO(4).
\end{equation*}
This decomposition has the convenient property that $\mu^{-1}(\K_{\theta_1,\theta_2}) = I_2 \otimes \SU(2) \otimes \SU(2)$ are already single-qubit operations. Wei and Di use this method, taking $\A$ and $\B$ are taken to be appropriate tori for Cartan decomposition. They find explicit circuits expressing elements of $\mu^{-1}(\A)$ (involving 6 CNOTs) and $\mu^{-1}(\B)$ (involving 4). The magic matrix $\M$ can be written as a product of one CNOT and some single-qubit gates (cf. \eqref{eq:magic-circuit}), so in total their method incurs a cost of $1+4+6+4+1 = 16$ CNOT gates.
\end{example}  

For the rest of the paper, let $\psi_1, \psi_2 : \PSO(8) \to \PSO(8)$ denote conjugation by $\Delta_{\{1,2,5\}}$ and $\Delta_{\{6,7\}}$, and $\chi_i = \T^{-1} \circ \psi_i \circ \T$. Note that $\psi_1, \chi_1$ were called just $\psi,\chi$ in \secsymbol\ref{sec:real-circuits}. A simpler description of $\chi_2$ is possible:
\begin{equation*}
\chi_2(V) = \T^{-1}(\Delta_{\{6,7\}} \T(V)\Delta_{\{6,7\}}) = \T^{-1}(\Delta_{\{6,7\}})V \T^{-1}(\Delta_{\{6,7\}}),
\end{equation*}
so $\chi_2$ is conjugation by $\T^{-1}(\Delta_{\{6,7\}}) = i\sigma_x \otimes \sigma_x \otimes \sigma_y$. No such simplification is possible for $\chi_1$, since $\Delta_{\{1,2,5\}} \notin \PSO(8)$; as observed in \secsymbol\ref{sec:real-circuits}, $\chi_1$ is not even linear.

\begin{proposition} \label{prop:K12}
We have
\begin{equation*}
\K_{\chi_1} = \T^{-1}(\K_{\psi_1}) = \T^{-1}\PSO(125|34678) = \mu\P(\Sp(2) \otimes \SU(2))
\end{equation*}
and 
\begin{equation*}
\K_{\chi_1,\chi_2} = \T^{-1}(\K_{\psi_1,\psi_2}) = \T^{-1}\PSO(125|348|67)
\end{equation*}
\end{proposition}

\begin{proof}
The first part was Theorem~\ref{thm:key-subgroup}. The second part follows from Lemma~\ref{lem:intersection}:
\begin{align*}
\K_{\psi_1,\psi_2} &= (\PSO(125|34678) \cap \PSO(123458|67))_0\\
&= \PSO(125|348|67).
\end{align*}
\end{proof}

Because of Proposition~\ref{prop:K12}, our method amounts to applying triality, using two Cartan decompositions coming from subgroups conjugate to
\begin{equation*}
\PSO(8) \supseteq \P(\SO(3) \times \SO(5)) \supseteq \P(\SO(3) \times \SO(2) \times \SO(3)),
\end{equation*}
and applying inverse triality. The next lemma provides explicit quantum circuits for each factor that will appear in these decompositions.
\begin{lemma} \label{lem:main}
\hfill
\begin{enumerate}[(a)]
\item Let $\mu^{-1}(\mathcal{A})$ be the set of 3-qubit gates of the form
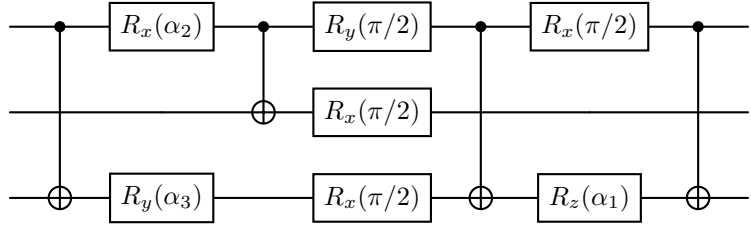
\begin{equation*}
\begin{quantikz}
& \ctrl{2} & \gate{R_x(\alpha_2)} &   \ctrl{1} & \gate{R_y(\pi/2)}  & \ctrl{2} & \gate{R_x(\pi/2)}   & \ctrl{2} &\\
&              &                    &    \targ{} & \gate{R_x(\pi/2)} &              &                                &              &\\
& \targ{}   &   \gate{R_y(\alpha_3)}&               & \gate{R_x(\pi/2)} & \targ{}  & \gate{R_z(\alpha_1)}              & \targ{}  &
\end{quantikz}
\end{equation*}
for $\alpha_1,\alpha_2,\alpha_3 \in \RR$. Then $\K_{\chi_1} \A \K_{\chi_1} = \PSO(8)$.

\item Let $\mu^{-1}(\mathcal{B})$ be the set of $3$-qubit gates of the form $R_y(\beta_1) \otimes R_y(\beta_2) \otimes I_2$ with $\beta_1, \beta_2 \in \RR$. Then $\K_{\chi_1,\chi_2}\B\K_{\chi_1,\chi_2} = \K_{\chi_1}$.

\item $\mu^{-1}(\K_{\chi_1,\chi_2})$ is the set of 3-qubit gates of the form 
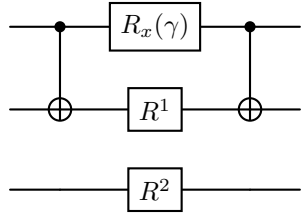
\begin{equation*}
\begin{quantikz}
 & \ctrl{1} & \gate{R_x(\gamma)} & \ctrl{1} &\\
& \targ{} & \gate{R^1}     & \targ{}  &\\
&         & \gate{R^2}     &          &
\end{quantikz}
\end{equation*}
with $\gamma \in \RR$ and $R^1,R^2 \in \SU(2)$.
\end{enumerate}
\end{lemma}
Actually, we have already proven (b) and (c): together they amount to the Cartan decomposition of $\Sp(2)$ used in the proof of Lemma~\ref{lem:real-gates-sp}. We give a second proof below using what we have learned about triality. It is also worth noting that the set $\A$ in (a) is \emph{not} a torus or even a group, so (a) is not expressing the ``standard'' Cartan decomposition as in Theorem~\ref{thm:cartan}. Rather, $\A$ was chosen to cover all possible canonical parameters with respect to $\K_{\chi_1}$ while still having a short circuit description.

\begin{proof}[Proof of Lemma~\ref{lem:main}] \hfill 
\begin{enumerate}[(a)]

\item An equivalent statement is that $\K_{\psi_1} \T(\A) \K_{\psi_1} = \PSO(8)$. Since $\K_{\psi_1} = \PSO(125|34678)$, it suffices by Proposition~\ref{prop:canonical-params} to show that for any $\sigma_1, \sigma_2, \sigma_3 \in [0,1]$ and any sign $s \in \{\pm 1\}$, there exists $A \in \T(\A)$ such that $A_{\{1,2,5\},\{1,2,5\}}$ has singular values $\sigma_1,\sigma_2,\sigma_3$ and determinant of sign $s$.

Write the circuit in (a) as the product of gates appearing in Figure~\ref{fig:triality}, i.e. 
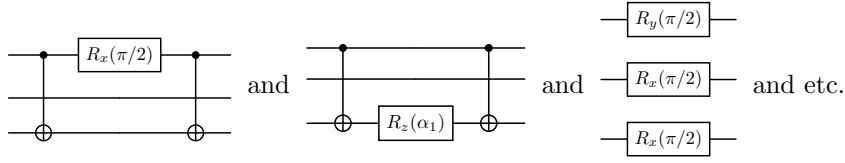
\begin{equation*}
\begin{tikzpicture} \node[scale=0.7] {\begin{quantikz}
 & \ctrl{2} & \gate{R_x(\pi/2)} & \ctrl{2} &\\
 &          & 				  &          &\\
 & \targ{}  & 				  & \targ{}  &
\end{quantikz}}; \end{tikzpicture} \, \raisebox{8mm}{\text{and}} \, 
\begin{tikzpicture} \node[scale=0.72] {\begin{quantikz}
 & \ctrl{2} & 				  & \ctrl{2} &\\
 &          & 				  &          &\\
 & \targ{} & \gate{R_z(\alpha_1)} & \targ{}  &
\end{quantikz}}; \end{tikzpicture} \,\raisebox{8mm}{\text{and}} \, 
\raisebox{-2mm}{\begin{tikzpicture} \node[scale=0.7] {\begin{quantikz}
 & \gate{R_y(\pi/2)} &\\
 & \gate{R_x(\pi/2)} &\\
 & \gate{R_x(\pi/2)} &
\end{quantikz}}; \end{tikzpicture}} \,\raisebox{8mm}{\text{and etc.}} \, 
\end{equation*}
Consulting Figure~\ref{fig:triality} to apply $\T \circ \mu$, we see $\T(\A)$ is the set of matrices
\begin{align*}
\begin{bmatrix}
-\cos \alpha_1 & 0 & 0 & 0 & 0 & 0 & \sin \alpha_1 & 0\\
0 & \cos \alpha_2 & 0 & -\sin \alpha_2 & 0 & 0 & 0 & 0\\
0 & 0 & 1 & 0 & 0 & 0 & 0 & 0\\
0 & 0 & 0 & 0 & \sin \alpha_3 & -\cos \alpha_3 & 0 & 0\\
0 & 0 & 0 & 0 & -\cos \alpha_3 & -\sin \alpha_3 & 0 & 0\\
\sin \alpha_1 & 0 & 0 & 0 & 0 & 0 & \cos \alpha_1 & 0\\
0 & -\sin \alpha_2 & 0 & -\cos \alpha_2 & 0 & 0 & 0 & 0\\
0 & 0 & 0 & 0 & 0 & 0 & 0 & 1
\end{bmatrix}
\end{align*}
The submatrix $\diag(-\cos \alpha_1, \cos \alpha_2, -\cos \alpha_3)$ in rows and columns $\{1,2,5\}$ obviously has the desired property.

\item  Using Figure~\ref{fig:triality} we compute $\T(\B) = \PSO(46|78|1|2|3|5)$.  We must prove that $\K_{\psi_1}$ and $\K_{\psi_1,\psi_2}\T(\B)\K_{\psi_1,\psi_2}$ are equal, which by Proposition~\ref{prop:K12} is equivalent to the equation
\begin{equation*}
\PSO(125|348|67) \cdot \PSO(1|2|3|46|5|78) \cdot \PSO(125|348|67) \overset{?}{=} \PSO(125|34678).
\end{equation*}
Rows and columns $1,2,5$ are uninteresting here, so let us remove them and relabel 6,7,3,4,8 as 1,2,3,4,5:
\begin{equation*}
\PSO(12|345) \cdot \PSO(14|25|3) \cdot \PSO(12|345) \overset{?}{=} \PSO(5).
\end{equation*}
or 
\begin{equation*}
\left[ \begin{smallmatrix}
\ast & \ast & 0 & 0 & 0\\
\ast & \ast & 0 & 0 & 0\\
0 & 0 & \ast & \ast & \ast\\
0 & 0 & \ast & \ast & \ast\\
0 & 0 & \ast & \ast & \ast
\end{smallmatrix}\right] \left[ \begin{smallmatrix}
\ast & 0 & 0 & \ast & 0\\
0 & \ast & 0 & 0 & \ast\\
0 & 0 & 1 & 0 & 0\\
\ast & 0 & 0 & \ast & 0\\
0 & \ast & 0 & 0 & \ast
\end{smallmatrix}\right] \left[ \begin{smallmatrix}
\ast & \ast & 0 & 0 & 0\\
\ast & \ast & 0 & 0 & 0\\
0 & 0 & \ast & \ast & \ast\\
0 & 0 & \ast & \ast & \ast\\
0 & 0 & \ast & \ast & \ast
\end{smallmatrix}\right] \overset{?}{=} \PSO(5).
\end{equation*}
This equation certainly holds: it is the standard Cartan decomposition for the pair $\P(\SO(2) \oplus \SO(3)) \subseteq \PSO(5)$, cf. Example~\ref{ex:SO-block}.

\item Applying $\T \circ \mu$ to the set of gates in question gives $\PSO(125|348|67)$ by Figure~\ref{fig:triality}, which is $\T(\K_{\chi_1,\chi_2})$ by Proposition~\ref{prop:K12}.
\end{enumerate}
\end{proof}

By assembling the circuits derived in Lemma~\ref{lem:main}, we arrive at our new decomposition for real 3-qubit gates.
\begin{theorem} \label{thm:main}
Every $V \in \SO(8)$ can be written as a product of at most 14 CNOT gates interleaved with single-qubit gates (elements of $\SU(2)^{\otimes 3}$). More explicitly, $V$ can be written as
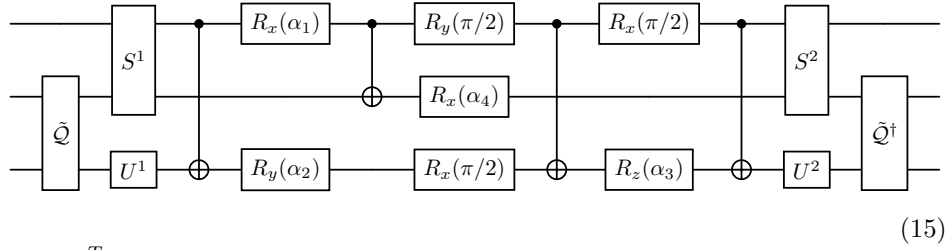
\begin{equation} \label{eq:main-decomp}
\begin{tikzpicture} \node[scale=0.85] {
\begin{quantikz}
 &            & \gate[2]{S^1} & \ctrl{2} & \gate{R_x(\alpha_1)} & \ctrl{1} & \gate{R_y(\pi/2)} &  \ctrl{2}& \gate{R_x(\pi/2)}    & \ctrl{2} & \gate[2]{S^2} & 				   &\\
 & \gate[2]{\tilde\Q} &               &           & 					            & \targ{}  & \gate{R_x(\alpha_4)} &          &                       &          & 			  & \gate[2]{\tilde\Q^\dagger} &\\
  &            & \gate{U^1}     &  \targ{}   & \gate{R_y(\alpha_2)} &          & \gate{R_x(\pi/2)} &  \targ{} & \gate{R_z(\alpha_3)} & \targ{}  & \gate{U^2}    & 				   &
\end{quantikz}};
\end{tikzpicture}
\end{equation}
where $S_1^T$ and $S_2$ have the form
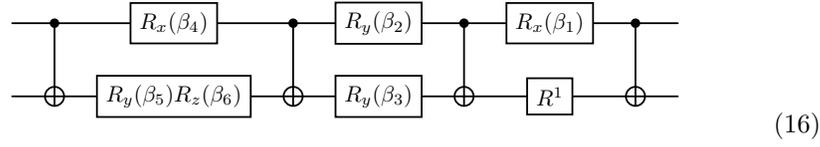
\begin{equation} \label{eq:main-decomp-sp}
\begin{tikzpicture} \node[scale=0.85] {\begin{quantikz}
& \ctrl{1} & \gate{R_x(\beta_4)} & \ctrl{1} & \gate{R_y(\beta_2)} & \ctrl{1} & \gate{R_x(\beta_1)} & \ctrl{1} &\\
& \targ{} & \gate{R_y(\beta_5)R_z(\beta_6)}          & \targ{}  & \gate{R_y(\beta_3)} & \targ{}  & \gate{R^1}          & \targ{}  &
\end{quantikz}}; \end{tikzpicture}
\end{equation}
and 
\begin{equation*}
\tilde\Q = \raisebox{-8mm}
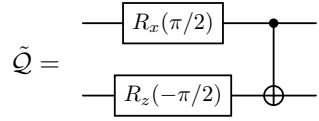
{\begin{tikzpicture} \node[scale=0.85] {\begin{quantikz}
  & \gate{R_x(\pi/2)} & \ctrl{1}& \\
  & \gate{R_z(-\pi/2)} & \targ{} &  
\end{quantikz}}; \end{tikzpicture}}
\end{equation*}
with $R^1,U^1,U^2 \in \SU(2)$ and $\beta_1,\ldots,\beta_6, \alpha_1, \ldots, \alpha_4 \in \RR$.
\end{theorem}
In this decomposition there are 14 CNOT gates total and 35 single-qubit rotations. Also, there are $6+4+6 = 16$ free angles above and 4 free elements of $\SU(2)$, for a total of $16 + 4\cdot 3 = 28$ free parameters. Since this is also the dimension of $\SO(8)$, the circuit is optimal in the sense that none of the parameters can be eliminated. Of course, this does not rule out the existence of shorter circuits.

\begin{proof} View $V$ as an element of $\PSO(8)$. This is necessary for use of the triality map, but irrelevant in the end: if we produce a circuit of the desired type for $V \in \PSO(8)$ which evaluates to $-V \in \SO(8)$, it can be easily modified to evaluate to $+V \in \SO(8)$, for instance by replacing $U^1$ with $-U^1$. 

 By (iterated) Cartan decomposition, every element of $\PSO(8)$ has the form 
\begin{equation*}
 = \M^\dagger \mu^{-1}(K_1)\mu^{-1}(B_1)\mu^{-1}(K_2)\mu^{-1}(A)\mu^{-1}(K_3)\mu^{-1}(B_2)\mu^{-1}(K_4)\M
\end{equation*}
with $K_1,\ldots,K_4 \in \K_{\theta_1,\theta_2}$ and $A \in \A$ and $B_1,B_2 \in \B$. Substituting in the circuit decompositions of these three sets from Lemma~\ref{lem:main} shows that every element of $\PSO(8)$ can be written as 
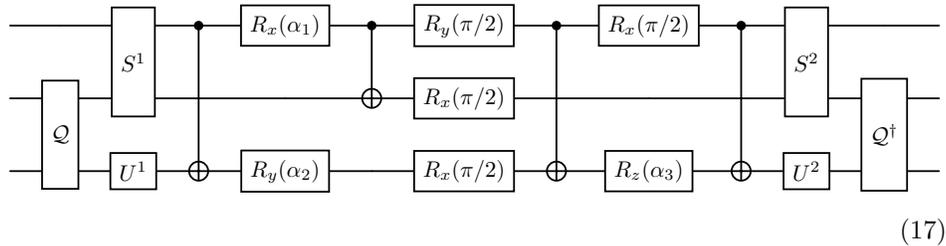
\begin{equation} \label{eq:main-decomp-2}
\begin{tikzpicture} \node[scale=0.85] {
\begin{quantikz}
 &            & \gate[2]{S^1} & \ctrl{2} & \gate{R_x(\alpha_1)} & \ctrl{1} & \gate{R_y(\pi/2)} &  \ctrl{2}& \gate{R_x(\pi/2)}    & \ctrl{2} & \gate[2]{S^2} & 				   &\\
 & \gate[2]{\Q} &               &           & 					            & \targ{}  & \gate{R_x(\pi/2)} &          &                       &          & 			  & \gate[2]{\Q^{\dagger}} &\\
  &            & \gate{U^1}     &  \targ{}   & \gate{R_y(\alpha_2)} &          & \gate{R_x(\pi/2)} &  \targ{} & \gate{R_z(\alpha_3)} & \targ{}  & \gate{U^2}    & 				   &
\end{quantikz}};
\end{tikzpicture}
\end{equation}
where $S_1^T$ and $S_2$ have the form
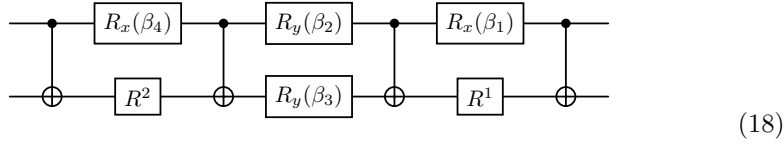
\begin{equation} \label{eq:main-decomp-sp-2}
\begin{tikzpicture} \node[scale=0.85] {\begin{quantikz}
& \ctrl{1} & \gate{R_x(\beta_4)} & \ctrl{1} & \gate{R_y(\beta_2)} & \ctrl{1} & \gate{R_x(\beta_1)} & \ctrl{1} &\\
& \targ{} & \gate{R^2}          & \targ{}  & \gate{R_y(\beta_3)} & \targ{}  & \gate{R^1}          & \targ{}  &
\end{quantikz}}; \end{tikzpicture}
\end{equation}
and
\begin{equation} \label{eq:main-magic-decomp}
\Q = \raisebox{-8mm}
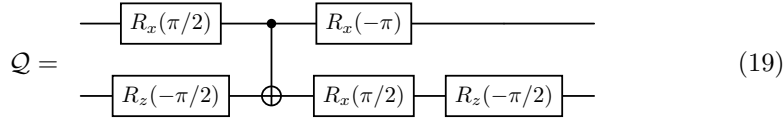
{\begin{tikzpicture} \node[scale=0.85] {\begin{quantikz}
  & \gate{R_x(\pi/2)} & \ctrl{1}& \gate{R_x(-\pi)} &   &\\
  & \gate{R_z(-\pi/2)} & \targ{} & \gate{R_x(\pi/2)} & \gate{R_z(-\pi/2)} & 
\end{quantikz}}; \end{tikzpicture}}
\end{equation}

This is almost the same as what we're trying to prove. First, the single-qubit rotations $R_x(\pi/2)$ and $R_z(-\pi/2)$ appearing on qubit 3 in the decomposition \eqref{eq:main-magic-decomp} of $I_2 \otimes \Q$ can be absorbed into the arbitrary 1-qubit gates $U^1, (U^2)^\dagger$ shown in \eqref{eq:main-decomp-2}. Similarly, the gate $R_x(-\pi)$ on qubit 2 in the same decomposition can be absorbed into the arbitrary 1-qubit gates $R^1, (R^2)^\dagger$; this requires the identity
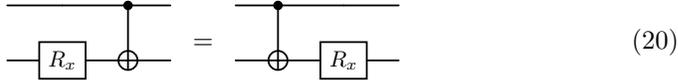
\begin{equation} \label{eq:CNOT-X}
\raisebox{-6mm}{\begin{tikzpicture} \node[scale=0.85] 
{\begin{quantikz} 
 &            & \ctrl{1} &\\
& \gate{R_x}  & \targ{} &
\end{quantikz}}; \end{tikzpicture}} = 
\raisebox{-6mm}{\begin{tikzpicture} \node[scale=0.85] 
{\begin{quantikz} 
 & \ctrl{1} &            & \\
  & \targ{} & \gate{R_x} &
\end{quantikz}}; \end{tikzpicture}}
\end{equation}
In this way we transform \eqref{eq:main-magic-decomp} into the circuit for $\tilde{\Q}$ shown in the theorem statement.

Next we use the identity \eqref{eq:CNOT-X} to move a few more gates around. In \eqref{eq:main-decomp-sp-2}, write the arbitrary 1-qubit gate $R^2 \in \SU(2)$ appearing in $S_2$ as $R_x(\beta_7)R_z(\beta_6)R_y(\beta_5)$, and write the analogous gate in $S_1^T$ as $R_x(\beta_7')R_z(\beta_6')R_y(\beta_5')$. By the identity \eqref{eq:CNOT-X} we may commute the $R_x$ factors into the very center of the circuit \eqref{eq:main-decomp-2} on qubit 2, replacing the gate $R_x(\pi/2)$ by $R_x(\beta_7)R_x(\pi/2)R_x(\beta_7')$. Setting $\alpha_4 = \beta_7+\pi/2+\beta_7'$ gives \eqref{eq:main-decomp} exactly.
\end{proof}

\section*{Acknowledgements}
I'm grateful to Jim van Meter for useful comments and for helping me navigate the literature in this area.

\bibliographystyle{plain}
\bibliography{Q}

\end{document}